\documentclass[pdflatex,sn-mathphys-num]{sn-jnl}% Math and Physical Sciences Numbered Reference Style

\usepackage{graphicx}%
\usepackage{multirow}%
\usepackage{amsmath,amssymb,amsfonts}%
\usepackage{amsthm}%
\usepackage[mathscr]{eucal}

\usepackage[normalem]{ulem}

\usepackage[title]{appendix}%
\usepackage{xcolor}%
\usepackage{textcomp}%
\usepackage{manyfoot}%
\usepackage{booktabs}%
\usepackage{algorithm}%
\usepackage{algorithmicx}%
\usepackage{algpseudocode}%
\usepackage{listings}%
\usepackage{csquotes}
\usepackage{enumitem}

\theoremstyle{thmstyleone}%
\newtheorem{theorem}{Theorem}%
\newtheorem{proposition}[theorem]{Proposition}
\newtheorem{lemma}[theorem]{Lemma}
\newtheorem{corollary}[theorem]{Corollary}
\newtheorem*{lemma*}{Lemma}
\newtheorem*{proposition*}{Proposition}

\theoremstyle{thmstyletwo}%
\newtheorem{remark}{Remark}%
\newtheorem{problem}{Problem}
\newtheorem{observation}{Observation}
\newtheorem{question}{Question}

\theoremstyle{thmstylethree}%
\newtheorem{definition}{Definition}%

\begin{document}
\title[Colorings of unrooted tree-based networks and related graphs]{Colorings of unrooted tree-based networks and related graphs}

\author[1]{\fnm{Mirko} \sur{Wilde}}\email{mirko.wilde@uni-greifswald.de}
\equalcont{These authors contributed equally to this work.}

\author*[1]{\fnm{Mareike} \sur{Fischer}}\email{mareike.fischer@uni-greifwald.de,email@mareikefischer.de}
\equalcont{These authors contributed equally to this work.}

\affil*[1]{\orgdiv{Institute of Mathematics and Computer Science}, \orgname{University of Greifswald}, \orgaddress{\street{Walther-Rathenau-Str. 47}, \postcode{17487} \city{Greifswald}, \state{Mecklenburg-Vorpommerania}, \country{Germany}}}

\abstract{In mathematical phylogenetics, evolutionary relationships are often represented by trees and networks. The latter are typically used whenever the relationships cannot be adequately described by a tree, which happens when reticulate evolutionary events happen, such as horizontal gene transfer or hybridization. But as such events are known to be relatively rare for most species, evolution is sometimes thought of as a process that can be represented by a tree with some additional edges, i.e., with a network that is still \enquote{somewhat tree-like}. In this context, different versions of tree-based networks have played a major role in recent phylogenetic literature. Yet, surprisingly little is known about their combinatorial and graph-theoretic properties. 
In our manuscript, we answer a recently published question concerning the colorability of a specific class of tree-based networks. In particular, we will investigate an even more general class of graphs and show their 3-colorability. This nicely links recent phylogenetic concepts with classical graph theory. Moreover, the ideas we use to answer the colorability question are new and might potentially be generalizable to other coloring problems in graph theory.}


\maketitle

\section{Introduction}

Understanding evolutionary relationships and their structures plays an important role in many research areas, e.g., concerning disease outbreaks \cite{plazzotta}, tumors \cite{Scott2020}, and even languages \cite{gray2013}. Mathematical phylogenetics is concerned with reconstructing these relationships from data (like DNA, RNA, or proteins), representing them with the help of graphs and analyzing properties of the latter. Traditionally, trees were used to depict evolutionary relationships \cite{SempleSteel2003}, but during the past two decades, biologists have become more and more interested in phylogenetic networks \cite{Huson2010}. This is due to the fact that trees are unsuitable to represent reticulate evolutionary events like horizontal gene transfer or hybridization, which cause cycles in the underlying graph. 

On the other hand, reticulate events are known to be relatively rare for most species \cite{linzPhD}, which is why networks with \enquote{too many cycles} are usually not of high interest. In fact, evolution is often assumed to be mainly tree-like with a few additional edges \cite{Francis2015b}. This point of view led to the introduction of tree-based networks, which have since played a big role in recent phylogenetic literature, and different types of tree-based networks have been distinguished. Basically, a tree-based phylogenetic network is simply a leaf-labelled connected graph which has a spanning tree with certain properties -- and depending on these properties, different types of such networks can be distinguished. We will elaborate this further in the course of this manuscript.

Due to the increasing interest in tree-based networks, unsurprisingly, various attempts have been made to understand their mathematical properties. For instance, it is known that tree-basedness can be decided in polynomial time for rooted networks \cite{Francis2015b}, while it is an NP-complete problem for unrooted networks \cite{Francis2018}. Moreover, some first attempts have been made to use classical graph theory to characterize certain sub-classes of tree-based networks \cite{Fischer2020}, as, for instance, edge-based networks, which are closely linked to series-parallel graphs \cite{Fischer2020,Fischer2023}. 

Yet only little is known about the colorability of tree-based networks, i.e., the question how many colors are at least needed to color the vertices of the network such that no two adjacent vertices have the same color -- which is the chromatic number of a graph. This is surprising on two levels: first, graph coloring is a very well established research topic in graph theory with plenty of results concerning the chromatic number of certain classes of graphs \cite{Diestel2017}. For instance, it is well-known that every tree is bipartite and can therefore be colored with at most two colors \cite{Diestel2017}. Second, certain types of coloring also play an important role in phylogenetics, e.g., leaf colorings and their extensions to all vertices (see, e.g., \cite{mikePhD}).\footnote{However, it has to be noted that colorings in phylogenetics typically have a different flavor than in graph theory as usually neighboring vertices are not forbidden to be in the same color -- in fact, often the leaf colors are given and the inner vertices of a given tree shall be colored such that the number of change edges, i.e., edges incident with two different colors, is minimized.} Therefore, it is surprising that not much is known about the colorability of certain types of tree-based networks. An exception is the manuscript by Hendriksen \cite{Hendriksen2018}, which gives interesting insight into the colorability of certain types of tree-based networks. However, this manuscript left the important question open as to whether a specific kind of such networks is always 3-colorable:

\begin{question}[adapted from Question 5.4, \cite{Hendriksen2018}]\label{qu:coloring}
Let $N$ be a strongly tree-based network. Is it true that $\chi(N) \leq 3$?
\end{question}

Note that the formal definition of the concepts required to understand Question \ref{qu:coloring} will be introduced subsequently. However, answering this question affirmatively is the main aim of our manuscript. In order to do so, we will use standard techniques from graph theory, such as Zykov operations, which we formally define and explain in the course of our manuscript.

Note that a positive answer to Question \ref{qu:coloring}  strengthens the close relationship between trees, which are known to be 2-colorable, and tree-based networks even more. However, as certain tree-based networks are known not to be 3-colorable \cite{Hendriksen2018}, the answer to Question \ref{qu:coloring} is not obvious. Moreover, as there are strongly tree-based networks which are not 2-colorable, the bound suggested by Question \ref{qu:coloring} is tight.  In our manuscript,  we will answer Question \ref{qu:coloring} on the 3-colorability of a certain type of these networks affirmatively, thus linking these networks more closely to trees. However, our main result does not only cover the aspired tree-based networks, but an even broader class of graphs which can be regarded as a generalization thereof. In particular, our approach uses similar ideas to the Zykov algorithms (cf.  \cite{McDiarmid1979}, \cite{Zykov1949}), which are well-known in graph theory but, to the best of our knowledge, have not yet been applied in mathematical phylogenetics. Perhaps more importantly, we are not aware of any instance in graph theory where Zykov operations have been used explicitly bound the chromatic number of a class of graphs thereby solving a theoretical question -- most research has been concerned with speeding up practical implementations of heuristic algorithms to compute this number for a given graph \cite{Mehrotra1996,Held2012,Malaguti2011,Brand2026}. We therefore believe that our ideas are not limited to phylogenetics but might also lead to new approaches in graph theory.

Along the way, we derive some results based on greedy procedures, such as a generalization of a result by Hendriksen as well as a more intuitive proof of another one of his results. However, while the greedy approach only leads to $\chi(N)\leq 4$, it does not appear to be sufficient to lower this bound to $\chi(N)\leq 3$. Instead, answering Question \ref{qu:coloring} requires the described machinery from graph theory.

Our manuscript is structured as follows: In Section \ref{sec:prelim}, we state all required definitions and relevant preliminary results known from the literature. In Section \ref{sec:results}, we present our own results. Finally, Section \ref{sec:disc} discusses and summarizes our findings and gives an outlook on future research.

\section{Preliminaries}\label{sec:prelim}

In this section, we present all relevant concepts and known results needed to derive our own results later on.

\subsection{Definitions and notions}

In this subsection, we specify the terminology and definitions used throughout the manuscript. In most of the cases, we follow the standard terminology of  Diestel \cite{Diestel2017}. 

\subsubsection*{Basic concepts from graph theory}

We define a graph $G=(V,E)$ to be a pair with a finite set $V$ and with $E\subseteq \binom{V}{2}$. In particular that means we will investigate undirected graphs. The elements of $V$ are called \emph{vertices} and the elements of $E$ are called \emph{edges}. In order to avoid ambiguity when several graphs are considered, we often write $V(G)$ instead of $V$ and $E(G)$ instead of $E$ in order to highlight the reference graph. If $V'\subseteq V,$ and $E' \subseteq \binom{V'}{2} \cap E$, then $H = (V',E')$ is called a \emph{subgraph} of $G$, and we write $H \subseteq G$. We define a \emph{subgraph induced by $W$} by the following construction: For $W\subseteq V$ we define $G[W] = (W, E_W)$ as the graph with vertex set $W$ and edge set $E_W = \{e\in E: e \in \binom{W}{2}\}$.  For $v\in V$ we define $G - v$ as $G[V\setminus \{v\}]$. For $e\in E$ we define $G-e$ as the graph $(V, E\setminus \{e\})$. 

If $G$ is a graph, $k\in \mathbb{N}$ and $f: V\rightarrow \{1, \ldots, k\}$ is a function, we call $f$ a \emph{$k$-coloring} of $G$ if $f(x) \neq f(y)$ for each $\{x,y\}\in E$. The \emph{chromatic number} $\chi(G)$ of a given graph $G$ is defined as the minimum $k\in \mathbb{N}$ for which a $k$-coloring of $G$ exists.

If $v\in V(G)$, then let $\deg_G(v)= |\{e\in E(G): v\in e\}|$. Then, $\deg_G(v)$ is called the \emph{degree of $v$ in $G$}. If $X,Y\subseteq V(G)$, then $E_G(X,Y)= \{\{x,y\}\in E(G): x\in X, y\in Y\}$. If $X\subseteq V(G)$, then the \emph{neighborhood} of $X$ in $G$, denoted by $N_G(X)$, is defined as $N_G(X) = \{v\in V(G)\setminus X: \{v,x\}\in E(G)$ $\mbox{ for some } x \in X\}$.
If $X =\{v\}$ for some $v \in V(G)$, for brevity, we often write $N_G(v)$ instead of $N_G(\{v\})$ to denote this set. By $K_n$ we denote the \emph{complete graph} on $n$ vertices defined as the graph with vertex set $V = \{1, \ldots, n\}$ and edge set $E = \binom{V}{2}$. In the case that $G$ is a graph and $K_n$ is isomorphic to a subgraph $H$ of $G$, we call $H$ a \emph{clique of $G$ of  size $n$} or shortly \emph{clique}. By $K_{n, m}$ we denote the \emph{complete bipartite graph} on $n, m$ defined as the graph with vertex set $V = A\cup B$ such that $A \cap B = \emptyset$, $|A|=n$,  $|B|=m$ and such that it has edge set $E = \{\{a,b\}: a\in A, b\in B\}$. 

Let $P = v_1, \ldots, v_k$ be a sequence of vertices in $V$ for some $k\in \mathbb{N}_{\geq 1}$. Assume that $\{v_i, v_{i+1}\}\in E$ for $i\in \{1, \ldots, k-1\}$. If each vertex is contained at most once in this sequence, then $P$ is called a \emph{path}. If $k\geq 3$, each vertex with the exception of $v_1$ is contained at most once in $P$ and $v_1=v_k$ is contained twice, then $P$ is called a \emph{cycle}. As usual we call a graph $G$ \emph{connected} if for any two vertices $u,v \in V(G)$ there exists a path $u,\ldots,v$ in $G$. A maximal connected subgraph of $G$ is called a \emph{connected component}. If $H$ is a maximal connected subgraph of $G$ such that $H-v$ is connected for each $v\in V(H)$, then $H$ is called a \emph{block}. A block $H$ with $|V(H)| = 2$ is called a \emph{bridge}. If $|V(H)| > 2$, then $H$ is called a \emph{non-trivial block}. Note that an edge $e\in E(G)$ is a bridge if and only if $G-e$ is disconnected and that each cycle is contained in some non-trivial block. Moreover, in each non-trivial block $H$ for each two edges $e_1, e_2\in E(H)$ there is some cycle in $H$ which contains both of them (cf. Lemma 3.1.3 in \cite{Diestel2017}). This fact implies that in each non-trivial block $H$ for each two vertices $v_1, v_2 \in E(H)$,  there is some cycle in $H$ which contains both of them (as $H$ is non-trivial, there are $e_1, e_2\in E(H)$ with $e_1\neq e_2$ such that $v_i\in e_i\ (i=1,2)$).

Finally, we want to introduce some concepts known as \emph{Zykov operations} in the literature \cite{Zykov1949,McDiarmid1979}.

\begin{definition}[Zykov operations] \label{def:zykov} Suppose that $x$ and $y$ are non-adjacent vertices in a graph $G$. Then we define the graphs $G^{\prime}_{xy}$ and $G^{\prime \prime}_{xy}$ as follows: \begin{itemize} \item $G^{\prime}_{xy}$ is obtained from $G$ by an \emph{addition}, i.e., by adding an edge joining $x$ and $y$.
\item $G^{\prime \prime}_{xy}$ is obtained from $G$ by an \emph{identification}, i.e., by replacing the vertices $x$ and $y$ by a single new vertex $v_{xy}$ and each edge $\{x,z\}$ or $\{y,z\}$ by an edge $\{v_{xy},z\}$ (cf. Figure \ref{fig:identify}). In order to avoid parallel edges, if both edges $\{x,z\}$,  $\{y,z\}$ are in $E(G)$, then only one edge $\{v_{xy},z\}$ is added. 
\end{itemize}
Both operations are called \emph{Zykov operations}.
\end{definition}

As we will elaborate later on, Zykov operations are well-known tools in graph theory to heuristically approximate the chromatic number of graphs, which is why they will also prove to be useful in order to tackle Question \ref{qu:coloring}. However, we first need to turn our attention to some more definitions.

\subsubsection*{Concepts concerning tree-based phylogenetic networks}

Next, we introduce some formal definitions from mathematical phylogenetics needed throughout this manuscript. Basically these are the same definitions and terms as used in \cite{Fischer2021}.\footnote{Note that Hendriksen's definitions \cite{Hendriksen2018} vary slightly -- e.g., the definition of tree-based networks used in \cite{Fischer2021} coincides with his definition of \emph{loosely} tree-based networks. This existing conflict in the literature stems from different approaches to generalize the binary setting, in which all inner vertices have degree 2, to the non-binary setting. We decided to stick with the notions used in \cite{Fischer2021} and adapted Hendriksen's question and definitions accordingly.} We start with phylogenetic trees and networks. 

\begin{definition}[Phylogenetic trees and networks] \label{def:net}
    Let $N$ be a simple connected graph and let $V^1$ be the set of all vertices in $N$ with degree $1$. We refer to $V^1$ as the leaves of $N$. We assume that no vertex in $N$ has degree $2$. Moreover, let the elements of $V^1$ be bijectively labelled by some label set $X$. Then,  $N$ is called an \emph{unrooted phylogenetic network} (or shortly \emph{network}) on $X$. Furthermore, we call the elements of $X$ the \emph{leaves} or \emph{taxa} of $N$. Also, if $N$ contains no cycles, we call it a \emph{phylogenetic tree}. 
\end{definition}

Next, we define tree-basedness. 

\begin{definition}[Tree-based networks]\label{def:treebased}
    Let $N$ be a phylogenetic network on $X$. We say that $N$ is \emph{tree-based} if it has a spanning tree $T$ with leaf set $X$.
\end{definition}

An example for a tree-based network is given by Figure \ref{fig:example} on the left. However, it is important to note that the existence of a spanning tree is guaranteed for all connected graphs \cite{Diestel2017}, so in particular for all phylogenetic networks. It is also guaranteed that all leaves of $N$ are leaves in every spanning tree $T$ of $N$. The crucial part of Definition \ref{def:treebased} is thus that $N$ is tree-based if it has a spanning tree $T$ that \emph{only} has $X$ as its leaf set. 

Now, the following concept is the most important concept of our manuscript: strongly tree-based networks. However, note that these networks are simply called tree-based networks in Hendriksen's manuscript \cite{Hendriksen2018}. 

\begin{definition}[Strongly tree-based networks] \label{def:treebased_hendriksen}
 Let $N$ be a network on $X$. We call $N$ \emph{strongly tree-based}, if there exists a spanning tree $T$ in $N$ with leaf set $X$  such that every edge in $N$ between vertices of degree $4$ or more is an edge in $T$ and every vertex of degree $2$ in $T$ is a vertex of degree $3$ in $N$.
 \end{definition}

Note that Definition \ref{def:treebased_hendriksen}  could be re-formulated along the lines of the following observation \cite[Theorem 2.5]{Hendriksen2018}: Each strongly tree-based network can be obtained by taking a tree $T$, adding degree-$2$ vertices to the edges (which are called points of attachments) and adding edges such that each new edge either joins two points of attachments or one point of attachment and one original vertex. This construction generalizes earlier constructions of tree-based networks \cite{Francis2015b, Jetten2018}. Moreover, it provides some intuition for Definition \ref{def:basis+representation} which specifies how to add certain edges to a given graph.

An example for a strongly tree-based network is given by Figure \ref{fig:example} on the right.

\begin{figure}
    \centering
\includegraphics{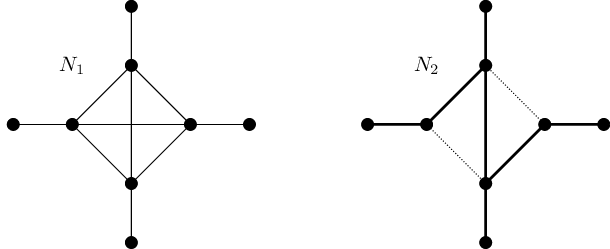}

\caption{Here $N_1$ is an example for a tree-based network which is not strongly tree-based. This can be easily seen by realizing that each spanning tree $T$ of $N_1$ cannot cover \emph{all} the edges of the $K_4$-subgraph of $N_1$ (as this would cause cycles), but this would be necessary for strong tree-basedness, because all of the inner vertices have degree 4 in $N_1$. On the other hand, considering the  spanning tree with edges highlighted in bold in the depiction of $N_2$, it can be verified that $N_2$ is strongly tree-based. Note that the same spanning tree can be used to verify that $N_1$ is tree-based, even if not strongly. }
\label{fig:example}
\end{figure}

As we formally state later on, our main aim is to prove the 3-colorability of strongly tree-based networks. However, before we can do that, we need some new definitions to prove our results.

\subsubsection*{New definitions required to prove our main result}

Our goal is to answer Question \ref{qu:coloring} affirmatively. However, instead of considering only strongly tree-based networks, our proof considers the following more general concept:

\begin{definition}[$k$-basis and $k$-representation]\label{def:basis+representation}
Let $k\in \mathbb{N}$ and let $G=(V,E)$ be a graph such that $V=A\cup B$ with $A \cap B=\emptyset$.

\begin{enumerate}
\item \label{itm:represents} 
Let $\mathfrak{b}: A \rightarrow \mathcal{P}(B)$ be some arbitrary function, where $\mathcal{P}(B)$ is the power set of $B$.) Then, we define a new graph $H$ by setting $V(H) = V$ and $E(H) = E \cup \{\{a,z\}: a\in A, z\in \mathfrak{b}(a)\}$. We call $H$ the graph \emph{defined} by the quadruple $(G, A, B, \mathfrak{b})$
and say that $(G, A, B, \mathfrak{b})$ \emph{represents} $H$.

\item \label{itm:basis_first} Let $H$ be a graph and let $G$, $A$, $B$ and $\mathfrak{b}$ be such that $(G, A, B, \mathfrak{b})$ represents $H$. If additionally the following three conditions hold, we say that $G$ is a \emph{$k$-basis} and $\mathfrak{b}$ is a \emph{$k$-representation} (of $H$): 
\begin{enumerate}
    \item \label{itm:condition1} Let $Z$ be a non-trivial block in $G$. Then $A\cap V(Z) = \emptyset$.
    \item \label{itm:condition2} $\deg_G(a) \geq 2$ for each $a\in A$.
    \item \label{itm:condition3} $\deg_G(a) + |\mathfrak{b}(a)| \leq k$ for each $a\in A$.
\end{enumerate}

\item \label{itm:basis_second} Let $H$ be a graph such that there exists some $(G, A, B, \mathfrak{b})$ representing $H$ for which $G$ is a $k$-basis. Then we say that $H$ \emph{has} a $k$-basis.
\end{enumerate}
\end{definition}

Informally speaking, if $(G, A, B, \mathfrak{b})$ represents $H$, then the vertices of $H$ can be labelled $A$ or $B$ in such a way that $G$ is a subgraph of $H$, and all edges of $H$ that are not in $G$ connect an $A$-vertex with a $B$-vertex. Note that this already has the flavor of considering a subgraph and adding some specific edges to derive the full graph -- just as we want to derive a phylogenetic network $N$ by taking a spanning tree and adding specific edges. However, in order to illustrate the somewhat technical Definition \ref{def:basis+representation}, in Figure \ref{fig:basis+representation} we give an example of a graph $H$ represented by four quadruples $(G_i,A_i, B_i, \mathfrak{b}_i)$ for $i=1, \ldots, 4$. Each quadruple $(G_i, A_i, B_i, \mathfrak{b}_i)$ for $i\in \{1, 2, 3\}$ is chosen such that it violates exactly one of the conditions of Part \ref{itm:basis_first} of Definition \ref{def:basis+representation} for $k=3$. This shows that the conditions are independent of each other. Only $(G_4, A_4, B_4, \mathfrak{b}_4)$ fulfills all conditions for $k=3$.

\begin{figure}
    \centering
 \includegraphics[scale=.83]{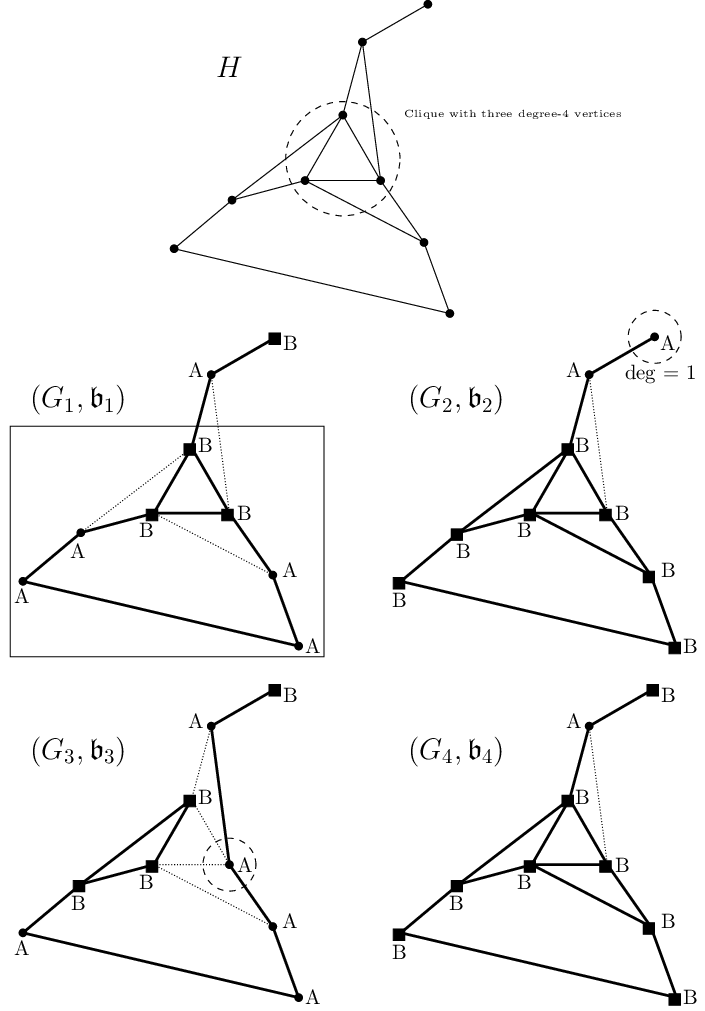}

\caption{Given is a graph $H$ (top) such that each quadruple $(G_i,A_i, B_i, \mathfrak{b}_i)$ represents $H$ for $i=1, \ldots, 4$ in the sense of Part \ref{itm:represents} of Definition \ref{def:basis+representation}. In the corresponding figures, $G_i$ is indicated by thick lines, the edges represented by $\mathfrak{b}_i$ are indicated by thin lines, and all $B$-vertices are represented by squares.
Only $(G_4, A_4, B_4, \mathfrak{b}_4)$ is a $3$-basis with $3$-representation, because it fulfills all conditions of Part \ref{itm:basis_first} of Definition \ref{def:basis+representation} for $k=3$. $(G_1, A_1, B_1, \mathfrak{b}_1)$ violates Condition \ref{itm:condition1} as there is a non-trivial block of $G$ containing vertices of $A$ (highlighted by the box). $(G_2, A_2, B_2, \mathfrak{b}_2)$ violates Condition \ref{itm:condition2}, because $A$ contains a vertex of degree 1 (highlighted by the circle). Finally, $(G_3, A_3, B_3, \mathfrak{b}_3)$ violates Condition \ref{itm:condition3} for $k=3$ as for the circled vertex $a\in A$ we have  $\deg_G(a)+|\mathfrak{b}(a)|=2+2=4>3$. However,  $(G_3, A_3, B_3, \mathfrak{b}_3)$ implies that $G_3$ is a $4$-basis of $H$ with $4$-representation $\mathfrak{b}_3$. Moreover, it should be noted that $H$ shows that not each graph which has a $3$-basis is a strongly-treebased network: $H$ contains a clique consisting of three vertices of degree $4$. By Definition \ref{def:treebased_hendriksen}, it follows that each spanning tree of $H$ which would verify that $H$ is a strongly tree-based network would contain all edges between vertices of degree $4$. As these form a cycle of length $3$, this would be a contradiction.
}
\label{fig:basis+representation}
\end{figure}

As stated above, Definition \ref{def:basis+representation} is motivated by the same intuition as Definition \ref{def:treebased_hendriksen}.  In order to emphasize this relationship, we now already state the following proposition,  which will be proven subsequently.

\begin{proposition}\label{prop:treebased_network}
Every strongly tree-based network has a $3$-basis $G$ with $\chi(G) \leq 2$.
\end{proposition}

In order to illuminate Proposition \ref{prop:treebased_network}, we now present a special case in which it is particularly easy to  construct a $3$-basis. Therefore, consider a pair $N$, $T$ which consists of a strongly tree-based network and some underlying spanning tree $T$ which not only satisfies the conditions specified by Definition \ref{def:treebased_hendriksen}, but additionally the following condition: each edge of $N$ which is not contained in $T$ connects a degree-$3$ vertex of $N$ and another vertex of degree $4$ or more in $N$. Note that this implies that each edge between two vertices of degree $3$ in $N$ is also contained in $T$.

In this case, a $3$-representation can be constructed by assigning each vertex of degree $3$ in $N$ to $A$ and each other vertex to $B$. For each vertex $v\in A$ there is at most one edge $e\in E(N)\setminus E(T)$ such that $v$ is incident with $e$, and if there is such an $e$ then this edge joins $v$ with some element of $B$. Hence, we can define $\mathfrak{b}(v) = \emptyset$ if there is no such edge and $\mathfrak{b}(v) = \{w\}$ if $e=\{v,w\}$ is such an edge. Then it can be verified  that $(T,A,B,\mathfrak{b})$ represents $N$, naturally $\chi(T)\leq 2$ holds (as $T$ is a tree) and additionally conditions \ref{itm:condition1}, \ref{itm:condition2} and \ref{itm:condition3} of the second part of Definition \ref{def:basis+representation} are satisfied.\footnote{We skip a formal proof here as we later on present a more general construction which applies to all instances of Definition \ref{def:treebased_hendriksen}.} Thus, this leads to a $3$-basis of $N$ which is as stated in Proposition \ref{prop:treebased_network}.

An example which illustrates Proposition \ref{prop:treebased_network} as well as the described  construction is given by Figure \ref{fig:observation}. In our proof of Proposition \ref{prop:treebased_network} we will show that each spanning tree of a network which is as specified by  Definition \ref{def:treebased_hendriksen} is at the same time a $3$-basis of the network if combined with a fitting $3$-representation. Hence, the thick lines on the right side of the Figure \ref{fig:observation} represent not only a spanning tree with certain properties verifying that this example is a strongly tree-based network, but also a $3$-basis. Moreover, we note that the graph $H$ given in Figure \ref{fig:basis+representation} is an example for a graph which has a $3$-basis but is not a strongly tree-based network. Therefore, the class of graphs with a $3$-basis is a proper generalization of the class of strongly tree-based networks.

\par\vspace{0.5cm}

The main idea of our manuscript is as follows: We will see that Proposition \ref{prop:treebased_network} is the decisive link between graph theory and phylogenetics which enables us to link statements about graphs with $k$-basis with statements about strongly tree-based networks. To demonstrate the power of this link, we will first re-prove the first part of Theorem \ref{thm:chromatic_number} which is known from phylogenetic literature by proving the even stronger Theorem \ref{thm:generalizeThm6}. Then, we will answer Question \ref{qu:coloring}, which makes a statement about strongly tree-based networks, affirmatively by showing the more general Theorem \ref{thm:main}, which holds for a larger class of graphs. In both cases Proposition \ref{prop:treebased_network} turns out to be crucial as it shows that statements about strongly tree-based networks can be considered as special cases of more general statements on graphs with a $k$-basis.

\par\vspace{0.5cm}

With regard to the inequality of Definition \ref{def:basis+representation}, Part \ref{itm:basis_first}, we note that this implies $\deg_H(a) \leq k$. However, there is no equivalence between Condition \ref{itm:condition3} and the inequality $\deg_H(a) \leq k$. As an example, consider for some $k\geq 2$ the graph $H = K_{1, k}$, which is the star graph with internal vertex $v$ and leaf set $\{1, \ldots, k\}$. Let $G = K_{1,k}$, $A = \{v\}$, $B = \{1, \ldots, k\}$ and $\mathfrak{b}: A\rightarrow \mathcal{P}(B)$, $\mathfrak{b}(v) = B$. Then $K_{1,k}$ is represented by $(G, A, B, \mathfrak{b})$ (in fact, both graphs equal $K_{1,k}$). Moreover, Conditions \ref{itm:condition1} and \ref{itm:condition2} are fulfilled and we have $\deg_H(v) = k$, but Condition \ref{itm:condition3} is violated, as $\deg_G(v)+|\mathfrak{b}(v)|=k+|B|\geq k+1$. Note that this example in particular shows that in Part \ref{itm:represents} of Definition \ref{def:basis+representation}, as $E(H)$ is a set rather than a multiset, it is possible that $\mathfrak{b}$ does not contribute additional edges that are not already contained in $G$.

Next, we turn our attention to known results relevant for our manuscript.

\subsection{Prior results}

In this section, we state results known from the literature that are important for the present manuscript. The most important such result is the following theorem, which shows that strongly tree-based networks can always be 4-colored.

\begin{theorem} [adapted from Theorem 4.1 in \cite{Hendriksen2018}] \label{thm:chromatic_number}
    Let $N$ be a \emph{strongly tree-based} network. Then $\chi(N) \leq 4$.
    
    On the other hand, for each $k\in \mathbb{N}$ there is a \emph{tree-based} network with $\chi(N) \geq k$.
\end{theorem}

Note that the fact that strongly tree-based networks can always be 4-colored only leaves a very small gap concerning Question \ref{qu:coloring}, which asks if a 3-coloring is always possible. This is due to the fact that it is known from the literature that trees with at least two vertices require always 2 colors \cite{Diestel2017} as they are bipartite. This easily leads to the following observation.

\begin{observation}\label{obs} There are strongly tree-based phylogenetic networks which are not 2-colorable.
\end{observation}

To verify that Observation \ref{obs} is correct, consider the network given by Figure \ref{fig:observation}. This network can be easily seen to be strongly tree-based. However, as it contains a triangle, the network cannot be 2-colorable. Thus, the affirmative answer we will give later on to Question \ref{qu:coloring} closes the gap (as we know 4 colors are always enough by Theorem \ref{thm:chromatic_number}, but 2 colors are not always enough as can be seen in Figure \ref{fig:observation} -- so the bound of 3 implied by our affirmative answer to Question \ref{qu:coloring} is tight). It also shows that strongly tree-based networks are even more closely related to trees than general tree-based networks.

\begin{figure}
    \centering
\includegraphics{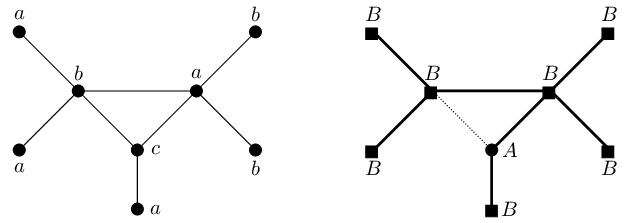}

\caption{[adapted from Figure 2(ii) in \cite{Hendriksen2018}]  A strongly tree-based network together with the spanning tree defining this property. The edges in the spanning tree are highlighted with thicker lines, whereas the only network edge that is not contained in the spanning tree is drawn as a thinner line. Note that this network is not 2-colorable due to the triangle in the middle. Thus, the depicted 3-coloring with colors $a$, $b$ and $c$ is minimal. Moreover, this example can be used to illustrate Proposition \ref{prop:treebased_network} which states that each strongly tree-based network has a $3$-basis. Here, the $3$-basis is given by the highlighted spanning-tree. The $3$-representation is given by $(T,A,B,\mathfrak{b})$, where $T$ is the spanning-tree, $A$ and $B$ are as indicated on the right and $\mathfrak{b}: A\rightarrow B$ is represented by the thin line.}
\label{fig:observation}
\end{figure}

Another result of \cite{Hendriksen2018} which we will refer to in our manuscript is the following theorem.

\begin{theorem}[adapted from Theorem 4.3 in \cite{Hendriksen2018}]\label{thm:strictly_tb} If a strongly tree-based network $N$ with leaf set $X$ has the additional property that it contains a spanning tree $T$ with leaf set $X$ that contains all edges of $N$ incident to vertices of degree at least 4, then $\chi(N)\leq 3$. 
\end{theorem}

Before we continue, we note that the specific strongly tree-based networks considered by the theorem are referred to as \emph{strictly tree-based} by Hendriksen.

Next, we need to recall two more statements from the literature. The following lemma is very easy to prove and therefore mentioned in the literature as an exercise or obvious statement (cf. \cite[V, Exercise 47]{Bollobas1998} and \cite[p. 141]{Harary1969}).  

\begin{lemma}\label{lem:blockcolors}
Let $k\geq 2$ and let $G$ be a graph. Then $\chi(G)\leq k$ if and only if $\chi(Z)\leq k$ for every block $Z$ in $G$.
\end{lemma}

As a corollary of this lemma, we directly obtain the following statement:

\begin{corollary}\label{cor:blockcolors}
Let $G$ be a graph such that each block contains at most $k$ vertices. Then, we have $\chi(G) \leq k$.     
\end{corollary}

Before we can continue with our own results, we need to consider some heuristic for the Vertex Coloring Problem, which consists in finding an optimal coloring of a graph. This heuristic is based on Zykov operations as introduced above. The idea here is merely that, by using these simple operations, one can obtain bounds on the chromatic number of a graph. In order to explain this more in-depth, we require one more definition and lemma.

\begin{definition}\label{def:zykov2}
    Let $G$ be a graph, $k\in \mathbb{N}_{\geq 1}$, $x,y$ a pair of non-adjacent vertices and $f: V(G''_{xy}) \rightarrow \{1, \ldots, k\}$ be some arbitrary function. Then we define $\widehat{f}: V(G) \rightarrow \{1, \ldots, k\}$ in the following way:

\begin{equation*} 
    \widehat{f}(v)=\begin{cases} f(v) & v\in V(G) \setminus \{x,y\},\\ f(v_{xy})  &v \in \{x,y\}.\end{cases}
\end{equation*}

\end{definition}

Using this definition, we now state the following simple lemma, which makes it possible to derive the chromatic number of a graph $G$ from the chromatic numbers of the graphs $G'_{xy}$ and $G''_{xy}$ of Definition \ref{def:zykov}, and which can easily be verified.

\begin{lemma}\label{lem:zykov2}
Let $G$ be a graph and $k\in \mathbb{N}_{\geq 1}$. Moreover, let $x,y$ be a pair of non-adjacent vertices and $g: V(G) \rightarrow \{1, \ldots, k\}$ be some arbitrary function. 
\begin{itemize}
    \item If $g(x) \neq g(y)$, then $g$ is a $k$-coloring of $G$ if and only if $g$ is a $k$-coloring of $G'_{xy}$.
    \item If $g(x) = g(y)$, then $g$ is a $k$-coloring of $G$ if and only if $g=\widehat{f}$ for a $k$-coloring $f: V(G''_{xy}) \rightarrow \{1, \ldots, k\}$ of $G''_{xy}$.
\end{itemize}    
\end{lemma}

This lemma is motivated by the graph-theoretical literature in which it is used to develop a heuristic for solving the vertex coloring problem \cite{McDiarmid1979,Zykov1949,Gualandi2012}. The basic idea of this heuristic is as follows: Using Zykov operations it is possible to construct a search tree which has the input graph as its root and complete graphs as its leaves. This heuristic is based on the fact that the Vertex Coloring Problem is trivially solvable for the leaves of this search tree (as for complete graphs the chromatic number is identical to the number of vertices), and hence -- iteratively using Lemma \ref{lem:zykov2} -- an optimal coloring of the root can be derived by investigating the whole search tree. As the Vertex Coloring Problem is known to be NP-complete \cite{Karp1972}, it is hardly surprising that this search tree will in almost all cases have a size which grows faster than exponentially in terms of the input length as shown in \cite{McDiarmid1979}. Hence, in practical implementations Zykov operations are often combined with other methods as brand-and-bound algorithms \cite{Mehrotra1996,Held2012} or SAT solvers \cite{Brand2026} in order to enhance the efficiency of these heuristics such that it may not be necessary to actually construct the whole search tree.

However, our approach to solving Question \ref{qu:coloring} will demonstrate the usefulness of Zykov operations not only for practical computation but in order to solve a question of theoretical interest. In our approach, we do not consider the whole search tree which is built by the heuristic, but only a subpath lying within this tree which has the property that it consists completely of graphs contained in the graph class which is defined by Definition \ref{def:basis+representation}. As we do not consider the whole search tree for our proofs, it will suffice to use the second part of Lemma \ref{lem:zykov2}, but the classic heuristic which constructs a complete search tree will need both kinds of Zykov operations, which is why we mention it in our manuscript for completeness. In Section \ref{sec:disc}, we discuss how our approach might be adjusted to solve different  coloring problems.

We are now finally in the position to turn our attention to our own results.

\section{Results}\label{sec:results}

\subsection{Greedy approaches}

In this section, we reconsider Theorems \ref{thm:chromatic_number} and \ref{thm:strictly_tb} by Hendriksen \cite{Hendriksen2018}. Our first aim is to show that a greedy strategy similar to the one used by Hendriksen in the proof for  Theorems \ref{thm:chromatic_number} can be used to prove the following result, which is a lot more general.

\begin{theorem}\label{thm:generalizeThm6} Let $k \in \mathbb{N}$ and let $H$ be a graph which is represented by $(G, A, B, \mathfrak{b})$ such that additionally Condition \ref{itm:condition3} of Definition \ref{def:basis+representation} holds, i.e., such that $\deg_G(a) + |\mathfrak{b}(a)| \leq k$ for each $a\in A$.  Then, if additionally we have $\chi(G) \leq k+1$, then we also have $\chi(H) \leq k+1$.
\end{theorem}

\begin{proof}
    Let $f: V(G) \rightarrow \{1, \ldots, k+1\}$ be a coloring of $G$. We assume $A\neq \emptyset$ (otherwise the statement is trivial). Let $A = \{a_1, \ldots, a_m\}$ be an enumeration of $A$. For each $i\in \{0, \ldots, m\}$ we define $H_i$ as the graph with vertex set $V(G)$ and edge set $E(G) \cup \bigcup_{j=1}^i \{\{a_j,b\}: b\in \mathfrak{b}(a_j)\}$. Note that $H_0 = G$ and $H_m = H$. 

    We construct a sequence $f_0 , f_1, \ldots, f_m$ beginning with $f_0 = f$ such that $f_i$ is a coloring of $H_i$ using at most $k+1$ colors. Assume that $f_i$ is already constructed for some $0\leq i < m$. We show how to construct $f_{i+1}$. If $f_i$ is a $(k+1)$-coloring of $H_{i+1}$, we simply set $f_{i+1}=f_i$. Otherwise,  there are some edges violating the condition that $f_i$ is a $(k+1)$-coloring of $H_{i+1}$. But such edges can only occur in $E(H_{i+1})\setminus E(H_i) = \{\{a_{i+1}, b\}: b\in \mathfrak{b}(a_{i+1})\}$. However, Condition \ref{itm:condition3} guarantees that $\deg_{H_{i+1}}(a_{i+1}) \leq k$. This implies $\{1, \ldots, k+1\}\setminus f_i(N_{H_{i+1}}(a_{i+1})) \neq \emptyset$. Hence, we can change the color of $a_{i+1}$ such that we obtain a coloring $f_{i+1}$ of $H_{i+1}$ using at most $k+1$ colors, which completes the proof.
\end{proof}

\begin{remark} \label{rem:N_with_basis} Before we continue, we note that the correctness of the first part of Theorem \ref{thm:chromatic_number} is a direct consequence of Theorem \ref{thm:generalizeThm6} combined with Proposition \ref{prop:treebased_network}: Proposition \ref{prop:treebased_network} implies that each strongly tree-based network fulfills all conditions of Definition \ref{def:basis+representation} for $k=3$, which in particular includes Condition \ref{itm:condition3}. Moreover, by the same proposition it is ensured that each strongly tree-based network has a $3$-basis $G$ with $\chi(G)\leq 2$, which is even stronger than $\chi(G)\leq 4$. Hence, Theorem \ref{thm:generalizeThm6} can be applied to each strongly tree-based network which shows that the first part of Theorem \ref{thm:chromatic_number} indeed holds.
\end{remark}

Before we continue with our main results, we turn our attention to a special class of strongly tree-based networks, namely the ones considered in Theorem \ref{thm:strictly_tb}. The proof of this theorem given by Hendriksen  \cite{Hendriksen2018} is somewhat technical -- it uses an inductive argument concerning the \emph{level} of a network, which is a measure of how much a network deviates from being a tree. Motivated by the success of greedy coloring strategies to prove Theorem \ref{thm:chromatic_number} and Theorem \ref{thm:generalizeThm6}, we provide an alternative proof for Theorem \ref{thm:strictly_tb} based on an even simpler greedy argument. This is possible due to the additional restrictions on the spanning tree $T$ of $N$ made by the theorem.

\begin{proof}[Alternative proof of Theorem \ref{thm:strictly_tb}] 

We consider $N$ and $T$ as stated in the theorem. Now we enumerate the vertices of $T$ as follows: We start with a leaf and call it $v_1$. We then take a vertex which is a neighbor of $v_1$ in $T$ and assign it $v_2$ and so forth, until we reach another leaf. Thus, in this first step, we have enumerated a leaf-to-leaf path with $v_1,v_2,\ldots,v_k$ for some $k\in \mathbb{N}$. Then, as long as there are still unnumbered vertices, we pick an unnumbered vertex which is adjacent in $T$ to a numbered one (which must exist as $T$ is connected), and continue the enumeration from there along a path to some unnumbered leaf. 

Once all vertices are numbered, we can consider each edge of $N$ as directed such that it points from the lower to the higher number. Now we consider a color set with three colors and go along our vertices according to our numbering. At every vertex, we assign one of the available colors that is not yet used in the neighborhood of the present vertex (but note that there might still be uncolored neighbors of this vertex). Clearly, if there is no vertex in $N$ with in-degree larger than 2, there will always be one of the three colors still available. So this is precisely what we will show now.

First note that within $T$, except for vertex $v_1$, which has in-degree 0 and out-degree 1 by construction, every vertex has in-degree 1. So if there was a vertex $v$ with in-degree at least 3 in $N$, this would hence imply that there are at least two edges of $E(N)\setminus E(T)$ incident with $v$. We now show that this cannot be the case. To see this, first note that $v$ cannot be a leaf: the leaf sets of $T$ and $N$ coincide by definition, and as $T$ is connected, the unique edge leading to each leaf in $N$ is also contained in $T$. Thus, no leaf is incident to an edge of $N$ that is not in $T$. So $v$ must be an inner node of $T$. Thus, it has a total degree of at least 2 in $T$. If it is also incident with two edges of $N$ which are not contained in $T$, this shows $\deg_N(v)\geq 4$. However, by the assumptions made on $N$ by the theorem, this would imply that \emph{all} edges incident with $v$ are contained in $T$, a contradiction to the assumption that $v$ is incident with two edges of $E(N)\setminus E(T)$. This completes the proof. \end{proof}

\begin{remark}\label{rem:implication_corollary 17}
    We note that Theorem \ref{thm:strictly_tb} is directly implied by Corollary \ref{cor:treebased_network} which will be proved later.
\end{remark}

\subsection{Main results}

It is our main aim to answer Question \ref{qu:coloring} affirmatively. However, we will actually prove an even stronger result, namely the following theorem, which is the main result of our manuscript and which delivers a better bound than Theorem \ref{thm:generalizeThm6}, albeit by presupposing more restrictive requirements. While Theorem \ref{thm:generalizeThm6} implies Hendriksen's original result on strongly tree-based networks (cf. the first part of Theorem \ref{thm:chromatic_number} and Remark \ref{rem:implication_corollary 17}), Theorem \ref{thm:main} will later on turn out to be the basis for answering Question \ref{qu:coloring} affirmatively.

\begin{theorem}\label{thm:main}
Let $k\geq 2$ and let $H$ be a graph with $k$-basis $G$ such that $\chi(G)\leq k$. Then $\chi(H) \leq k$.
\end{theorem}

In order to see that Theorem \ref{thm:main} provides the desired answer to Question \ref{qu:coloring}, we need to make a connection between strongly tree-based networks and graphs with a $k$-basis. This is done by Proposition \ref{prop:treebased_network} which has been stated earlier.

Note that Proposition \ref{prop:treebased_network} is actually slightly stronger than needed to fulfill the conditions of Theorem \ref{thm:main}, as $\chi(G)\leq 3$ would be sufficient, but we will prove the stronger statement that we have $\chi(G)\leq 2$. Finally, from Theorem \ref{thm:main} together with Proposition \ref{prop:treebased_network} we then easily derive the following statement as a direct corollary, which answers Question \ref{qu:coloring} affirmatively:

\begin{corollary}\label{cor:treebased_network}
    Let $N$ be a strongly tree-based network. Then $\chi(N) \leq 3$.
\end{corollary}

So it remains to prove Theorem \ref{thm:main} and Proposition \ref{prop:treebased_network}. However, before we do so, we first need to prove some preliminary statements that will be used later on.

\subsection{Preliminary results needed to prove our main result}

In this section, we prove some preliminary statements that will be used to prove our main result, which is Theorem \ref{thm:main}. As a first step, we consider again 
Definition \ref{def:basis+representation}, which is highly technical. While its connection to strongly tree-based networks will be elaborated later, we can already state here that it will serve as a generalization of Definition \ref{def:treebased_hendriksen}. Loosely speaking, just as strongly tree-based networks can be thought of as trees with certain additional edges, Definition \ref{def:basis+representation} allows us to think of general graphs as subgraphs with certain additional edges. 

Next, we state two lemmas that will elucidate Condition \ref{itm:condition1} of Definition \ref{def:basis+representation}.

\begin{lemma} \label{lem:bridges_in_basis}
   Let $H$ be a graph which is represented by the quadruple $(G, A, B, \mathfrak{b})$ such that Condition \ref{itm:condition1} of Definition \ref{def:basis+representation} is fulfilled, i.e., such that $A\cap V(Z)=\emptyset$ for all non-trivial blocks $Z$ of $G$. Let $v \in A$ and $e\in E(G)$ with $v\in e$. Then $e$ is a bridge in $G$.
\end{lemma}
\begin{proof}
   Seeking a contradiction, let us assume that $e$ is not a bridge. Then $e$ is contained in some $C$ which is a cycle in $G$. There is some non-trivial block $Z$ of $G$ containing $C$. Then $v\in e$ and $e\in C \subseteq Z$ imply $v\in Z$. But we have assumed $v\in A$, so this is in contradiction to Condition \ref{itm:condition1}. This shows that the assumption was wrong and thus completes the proof.
\end{proof}

This lemma has the following immediate consequence:

\begin{lemma}\label{lem:neighbor_of_component}
     Let $H$ be a graph and let $G$, $A$, $B$ and $\mathfrak{b}$ such that $H$ is represented by $(G, A, B, \mathfrak{b})$ and such that Condition \ref{itm:condition1} of Definition \ref{def:basis+representation} is fulfilled, i.e., such that $A\cap V(Z)=\emptyset$ for all non-trivial blocks $Z$ of $G$. Let $C$ be a connected component of $G[B]$ and $v\in N_G(V(C))$. Then $v$ is adjacent to exactly one vertex of $C$ in $G$.
\end{lemma}
\begin{proof}
    Seeking a contradiction, let us assume there are two different vertices $v_1, v_2\in C$ which are neighbors of $v$ in $G$. As $C$ is connected, there is a $v_1$-$v_2$-path $P$ in $C$. Then $P$ can be extended to $P'$ via $v_1vv_2$ such that $P'$ is a cycle in $G$ which contains $\{v,v_1\}$ and $\{v,v_2\}$. But this contradicts Lemma \ref{lem:bridges_in_basis} as $\{v,v_1\}$ and $\{v,v_2\}$ would not be bridges (as they are part of a cycle). This completes the proof.
\end{proof}

So Condition \ref{itm:condition1} tells us that $\mathfrak{b}$ imposes on $G$ a forest-like structure in the following sense: If we contract all connected components of $G[B]$ in $G$, then we obtain a forest $F$ with a natural bijection between $E(F)$ and $E(G) \setminus E(G[B])$. This already gives a first hint why Definition \ref{def:basis+representation} may generalize tree-like structures as the one described in Definition \ref{def:treebased_hendriksen}. Note that Proposition \ref{prop:treebased_network} verifies this observation as it shows that every strongly tree-based network indeed has a 3-basis. We are now finally in the position to prove this assertion.
 
\begin{proof}[Proof of Proposition \ref{prop:treebased_network}] 
    Let $N$ be a strongly tree-based network and let $T$ be as required by Definition \ref{def:treebased_hendriksen}. As $T$ is a tree, we have $\chi(T) \leq 2$. We now show that $T$ is in fact a 3-basis for $N$, and we need to find a $3$-representation $\mathfrak{b}: A\rightarrow \mathcal{P}(B)$. 
    
    Let $A^{\prime}$ be the set of all vertices in $T$ with degree $2$. (Note that $A'=\emptyset$ is possible.) Definition \ref{def:treebased_hendriksen} tells us that $\deg_N(v) = 3$ for each $v\in A^{\prime}$, which implies for each such $v$ the existence of exactly one edge $e_v\in E(N)\setminus E(T)$ with $v\in e_v$. Furthermore, we claim that every edge in $E(N)\setminus E(T)$ has the form $e_v$ for at least one vertex $v\in A^{\prime}$. To prove this claim we assume that there is some $e=\{x,y\}\in E(N)\setminus E(T)$ with $e\cap A' = \emptyset$. Then $2\notin \{\deg_T(x), \deg_T(y)\}$. If $1\in \{\deg_T(x), \deg_T(y)\}$, say $\deg_T(x) = 1$, then $x$ is a leaf in $T$. As we have assumed $T$ to be a tree as specified by  Definition \ref{def:treebased_hendriksen}, we conclude that $x\in X$ and $\deg_N(x) = 1$. But this implies $e\in E(T)$, which is a contradiction to our assumption. On the other hand, if $\deg_T(x) \geq 3$ and $\deg_T(y) \geq 3$, then, considering $e=\{x,y\} \in E(N)\setminus E(T)$, we must have $\deg_N(x)\geq 4$ and $\deg_N(y) \geq 4$ (as $N$ is connected), which is again a contradiction to Definition \ref{def:treebased_hendriksen}.
    
    The above contradictions show that indeed every edge in $E(N)\setminus E(T)$ contains a vertex $v \in A'$. In particular, if $A'=\emptyset$, we must have $E(N)=E(T)$ and thus also $N=T$.

    Next, we define $A$, $B$ and $\mathfrak{b}$, and we subsequently show that $(T,A,B,\mathfrak{b})$ represents $N$. Now, in order to define $\mathfrak{b}$, for each $e\in E(N)\setminus E(T)$ we fix one vertex $v_e \in A^{\prime}\cap e$ 
    (note that there might be two vertices to choose from). Set $A= \{v_e: e\in E(N)\setminus E(T)\}$, set $B= V(N)\setminus A$, and for every $v_e\in A$ set $\mathfrak{b}(v_e)= e\setminus \{v_e\}$. Note that $A\subseteq A'$. 
    
    To prove that $T$, $A$, $B$ and $\mathfrak{b}$ represent 
    $N$, it suffices to show $\mathfrak{b}(a) \subseteq B$ for each $a\in A$. Note that in particular, if $A'=\emptyset$, we also have $A=\emptyset$, so there is nothing to show. So now let $A\neq \emptyset$, and, seeking a contradiction, assume that $a = v_e$ for some $e\in E(N)\setminus E(T)$ and that $e=\{a,a'\}$ with $a' \in A$, i.e., $a' = v_f$ for some $f\in E(N) \setminus E(T)$. As $a'\in A \subseteq A'$, we know $\deg_T(a') = 2$. But $a'$ is also an endpoint of $e$ and $f$ which are both in $E(N)\setminus E(T)$. This implies $\deg_N(a') \geq 2+2=4$ (as $a'$ is incident with precisely two edges in $T$ and with at least two edges in $N\setminus T$),  which contradicts Definition \ref{def:treebased_hendriksen}, which states that $\deg_N(a')=3$. Thus, we know that $(T, A, B, \mathfrak{b})$ represents $N$. 
    
    Now it remains to show that $T$ is a 3-basis for $N$. Therefore, we need to check the conditions of Definition \ref{def:basis+representation}.
    
    As $T$ is a tree, it does not contain any non-trivial block, so Condition \ref{itm:condition1} of Definition \ref{def:basis+representation} holds trivially. As $A\subseteq A'$, it is true that $\deg_G(v) = 2$ for all $v\in A$, which shows that Condition \ref{itm:condition2} holds. Finally, we note that by definition we have $|\mathfrak{b}(v)| = 1$ for all $v\in A$, so we have $\deg_T(v)+|\mathfrak{b}(v)| = 2+1=3$ for all $v\in A$. This shows that $T$ is a $3$-basis for $N$ with $\chi(T)\leq 2$, which completes the proof.
    \end{proof}
    
\subsection{Proof of our main result}

The main aim of this section is to prove Theorem \ref{thm:main}. Note that if the set $A\subseteq V$ defining the $k$-basis $G$ of $H$  in the theorem is empty, we have $G=H$. In this case, the statement of the theorem is obvious (as clearly we have $\chi(G)=\chi(H)$ in this case). Therefore, the remainder of this paper will be devoted to a proof of Theorem \ref{thm:main} for the case $A\neq \emptyset$. 

The main idea of the proof is as follows: For each graph $H$ with $k$-basis $G$, $k$-representation $\mathfrak{b}$ as well as sets $A$ and $B$ as specified by Definition \ref{def:basis+representation}, we find a graph $H^-$ with $k$-basis $G^-$ and $k$-representation $\mathfrak{b}^-$ with $|A^-| < |A|$ such that $\chi(H^-) \leq k$ implies $\chi(H) \leq k$. Section \ref{sec:reduction} is concerned with the description of these reductions. Subsequently, Section \ref{sec:proof} will show how to use these reductions to prove Theorem \ref{thm:main} inductively.

\subsubsection{Reduction steps}
\label{sec:reduction}
In this section, we describe the reduction steps needed for our inductive proof of Theorem \ref{thm:main}. The idea is based on Zykov operations as introduced earlier. These operations will be an elementary part of the more complex constructions which we  describe in the following pages. Roughly speaking, the idea is that, given a graph as in Theorem \ref{thm:main}, we can reduce this graph by using appropriate Zykov operations to a smaller graph of a similar structure. 

We start with the following lemma, which describes a construction which shows that -- given some graph $H$ with $k$-basis $G$, $k$-representation $\mathfrak{b}: A\rightarrow \mathcal{P}(B)$ and $\chi(G)\leq k$ -- some additional assumptions regarding the structure of $G[B]$ can be made. Later on, we will in particular use Lemma \ref{lem:shrinking} to show that we can assume without loss of generality that each connected component of $G[B]$ is a clique with at most $k$ elements.

\begin{figure}
    \centering
\includegraphics[scale=0.9]{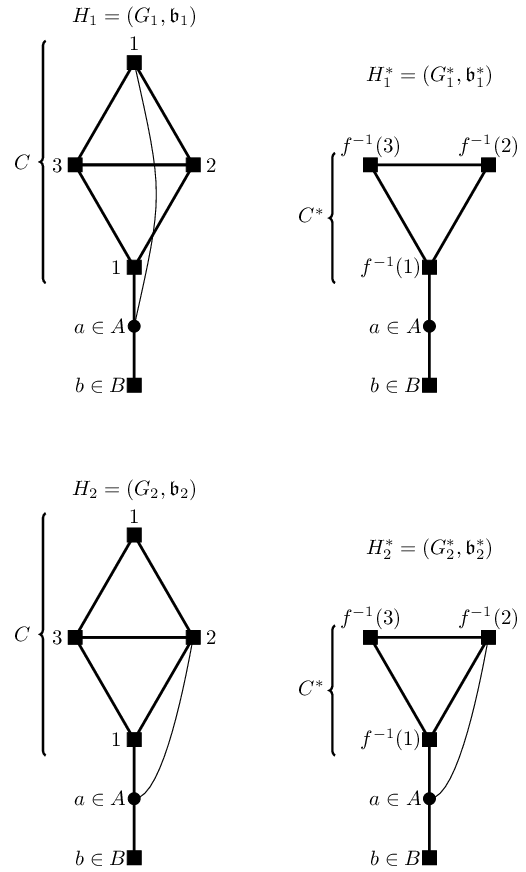}
\caption{Two examples which illustrate the construction of $\mathfrak{b}^*(a)$ as formally defined in Equation \eqref{eq:def_repr}. Each example consists of a pair $H_i, H^*_i (i=1,2)$ such that $H^*_i$ is obtained from $H_i$ by the construction which is described in Lemma \ref{lem:shrinking}. $H_i$ has $3$-basis $G_i$ and $3$-representation $\mathfrak{b}_i$, $H^*_i$ has $3$-basis $G^*_i$ and $3$-representation $\mathfrak{b}^*_i$. The edges belonging to the respective $3$-basis are depicted with thick lines in each graph, whereas the edges which are not contained in the $3$-basis are thin. For example, $H^*_2$ contains exactly one edge which is not contained in $G^*_2$, hence the graphical representation of $H^*_2$ contains exactly one thin line. Note that the construction of $G^*_i$ (for $i=1,2$) requires identifying vertices with identical colors and there is exactly one such pair in each example. Thus, here the modification effectively leads to the deletion of a single vertex in each example. Moreover, note that in the first example,  we have $\mathfrak{b}_1^*(a) \neq \emptyset$. Indeed, as $H^*_1$ is the graph represented by $G^*_1$ and $\mathfrak{b}^*_1$, by Definition \ref{def:basis+representation}, we have $E(H^*_1) = E(G^*_1) \cup \{\{a, v\}: v\in \mathfrak{b}^*_1(a)\}$. But in this case, all edges of the form $\{a,v\}$ with $v\in \mathfrak{b}^*_1(a)$ (of which there is only one) are already contained in $E(G^*_1)$, hence $H^*_1 = G^*_1$. On the other hand, as we have already noted, $H^*_2$ contains an additional edge with endpoint $a$ which is not already contained in $G^*_2$. This is due to the fact that the additional edge is not connected with one of the vertices colored in $G_2$ by color $1$ and in this example only vertices with color $1$ are identified while the other vertices are de facto left unmodified (despite being formally replaced by other vertices).} 
\label{fig:shrinking2}
\end{figure}

\begin{lemma}\label{lem:shrinking}
    Let $k \in \mathbb{N}_{\geq 2}$ and let $G$ be a connected graph with $\chi(G) \leq k$. Let $H$ be a graph with $k$-basis $G$ (and corresponding sets $A$, $B$ as in Definition \ref{def:basis+representation})  and $k$-representation $\mathfrak{b}: A\rightarrow \mathcal{P}(B)$. Let $C$ be a connected component of $G[B]$. Then there is a graph $H^{\ast}$ with $k$-basis $G^{\ast}$ and $k$-representation $\mathfrak{b}^{\ast}: A^{\ast} \rightarrow \mathcal{P}(B^{\ast})$ and there exists a set $W^\ast\subseteq V(G^\ast)$ such that the following statements hold, where we define  $C^\ast:=G^\ast[W^\ast]$:
\begin{enumerate}[label=(\roman*)]
    \item \label{itm:first} $V(G^{\ast}) = (V(G) \setminus V(C)) \cup W^{\ast}$ and  $W^\ast\cap V(C)=\emptyset$. 
    \item \label{itm:third} $G[V(G) \setminus V(C)] = G^\ast[V(G^\ast) \setminus V(C^\ast)]$.
    \item \label{itm:fourth} $C^{\ast}$ is a clique with at most $k$ vertices.
    \item \label{itm:seventh} If $G$ is connected, then $G^*$ is also connected.
    \item \label{itm:fifth} $A^{\ast} = A$, and for all $a\in A$ we have $\mathfrak{b}(a) \setminus V(C) = \mathfrak{b}^{\ast}(a)  \setminus V(C^{\ast})$.
    \item \label{itm:second} $C^{\ast}$ is a connected component of $G^{\ast}[B^{\ast}]$.
    \item \label{itm:sixth} $\chi(H) \leq \chi(H^{\ast})$.
\end{enumerate}
\end{lemma}
In Figure \ref{fig:shrinking2} we give two examples for the procedure which is described in Lemma \ref{lem:shrinking}. Each graph $H_i\ (i=1,2)$ has a $3$-representation $(G_i, \mathfrak{b}_i)$ with $\chi(G_i) = 3$. Moreover, in both cases $G_i[B]$ is obtained from $G_i$ by deleting $a$. Hence, in both cases $C$ is a connected component of $G[B]$, and in both examples the conditions of Lemma \ref{lem:shrinking} are satisfied. In each case, $H_i^*\ (i=1,2)$ results from the application of Lemma \ref{lem:shrinking} to $H_i$, by which $C$ is replaced with $C^*$.

\begin{proof}
We prove the statement by providing an explicit construction. Before we construct $H^\ast$, $A^\ast$, $B^\ast$ and $\mathfrak{b}^\ast$, we construct $G^\ast$  and $W^\ast$. Then, we construct $A^\ast$, $B^\ast$ and $\mathfrak{b}^\ast$ before explicitly stating $H^\ast$. This three step approach enables us to prove some of the statements first and to use them subsequently to prove the remaining statements.   \begin{itemize} 
\item We start by explaining our construction of $G^\ast$ and $W^\ast$. In the following, let $W = V(C)$. We consider some minimal $r$-coloring $g: W \rightarrow \{1, \ldots, r\}$ of $C$, i.e. $r = \chi(C)$ (for an example see the numbered vertices of $C$ in Figure \ref{fig:shrinking2}). Clearly, $r\leq \min(k, |W|)$ because $\chi(C)\leq \chi(G) \leq k$. From the minimality of $r$ it follows that $g$ is surjective (otherwise, we could delete the unused colors and proceed with a smaller value of $r$). Now we generate a graph $G^*$ by repeating the following operation on $G$: Start with $G^*=G$ and $C^*=C$. As long as there is a color $i\in \{1,\ldots,r\}$ that occurs several times in $C^*$, identify all vertices of color class $g^{-1}(i)$ for $i\in \{1, \ldots, r\}$ (which are non-adjacent) in $C^*$ (i.e., apply the second Zykov operation defined in Definition \ref{def:zykov}), in order to derive the updated version of $C^*$ and thus also the updated version of $G^*$. Note that by this process, the vertices of $C$ in $G$ are replaced by $r$ new vertices corresponding to the color classes of $g$, resulting in graph $G^*$. 

So we can assume that $V(G^*) = (V(G) \setminus V(C)) \cup \{g^{-1}(i): i\in \{1, \ldots, r\}\}$ such that each vertex in $V(G^*)\setminus V(G)$ is identical to a color class of $g$. However, note that at this stage of our construction all vertices of $G^*$ are uncolored, i.e., $g$ is only used for constructing $G^*$, but not for coloring it. Later we will show how to use the reduction from $G$ to $G^*$ to find for each coloring $f^*$ of $G^*$ (which may be completely unrelated to $g$) a corresponding coloring of $G$ which uses the same number of colors as $f^*$.

Now, let $W^{\ast} = \{g^{-1}(i): i\in \{1, \ldots, r\}\} = V(G^*)\setminus V(G)$. By definition of $r$ it is clear that $W^{\ast}$ contains at most $k$ elements. 

\item We can now already show that statement \ref{itm:first} of Lemma \ref{lem:shrinking} holds for our choice of $G^\ast$ and $W^\ast$. Recall that we have $V(G^\ast)=(V(G)\setminus V(C)) \cup W^\ast$, which shows the first part of \ref{itm:first}. For the second part, recall that $W=V(C)$ and note that $W\cap W^\ast=\emptyset$ (as all vertices of $W$ were removed to form the vertices of $W^\ast$). This completes the proof of Part \ref{itm:first}. 

\item Part \ref{itm:third} also holds for our choice of $G^\ast$ and $C^\ast$ as the Zykov operations we performed left $G[V(G)\setminus W]=G[V(G)\setminus V(C)]$ unchanged (we only changed vertices of $W$ to form $W^\ast$). Thus, \ref{itm:third} obviously holds, too. 

\item We now claim hat $W^{\ast}$ is a clique. Assume this is not the case. Then there exist  $v,w\in C^{\ast}$ with $v=g^{-1}(i)$, $w = g^{-1}(j)$ and $ i\neq j$ such that $\{v,w\}\notin E(C^*)$. As $v$ and $w$ were constructed using Zykov operations by identifying all vertices in $g^{-1}(i)$ or $g^{-1}(j)$, respectively, it follows that $E_G(g^{-1}(i), g^{-1}(j)) = \emptyset$. But then there is a coloring $h:V(C)\rightarrow \{1,\ldots, r\}\setminus \{j\}$ of $C$ with $h^{-1}(m) = g^{-1}(m)$ for $m\in \{1, \ldots, r\}\setminus \{i,j\}$ and $h^{-1}(i) = g^{-1}(i)\cup g^{-1}(j)$. As $h$ employs only  $r-1$ colors, this is a contradiction to the minimality of $g$. Hence, $W^{\ast}$ is a clique, which shows \ref{itm:fourth}.

\item Let $D_1, \ldots, D_m$ $(m\geq 1)$ be the connected components of $G[V(G)\setminus V(C)]$ and $V_i = V(D_i)$ be the vertex set of $D_i$. We note that $V(C) \cup \bigcup_{i=1}^m V_i$ is a partition of $V(G)$. By \ref{itm:third} it is also true that $\bigcup_{i=1}^m V_i$ is a partition of $V(G^\ast) \setminus V(C^*)]$. Hence, $V(C^*) \cup \bigcup_{i=1}^m V_i$ is a partition of $V(G^*)$. Next, we note that connectedness is an equivalence relation, i.e., if $x$ is connected to $y$ by a path and $y$ is connected to $z$ by a path, then $x$ is also connected to $z$ by a path. By \ref{itm:fourth} we know that $W^*$ is a clique. Thus, $C^*$ is a connected subgraph of $G^*$. Moreover, by \ref{itm:third} we conclude that $D_i = G[V_i] = G^*[V_i]$ for each $i\in \{1, \ldots, m\}$. By definition of $D_i$ it follows that $G^*[V_i]$ is a connected subgraph for each $i\in \{1, \ldots, m\}$. It remains to show that each $G^*[V_i]$ is connected to $C^*$ in $G^*$. In other words, we want to show that $E(V(C^*), V_i) \neq \emptyset$ for each $i\in \{1,\ldots, m\}$. However, this follows immediately if we recall that $C^*$ has been obtained from $C$ by identifying certain vertices of $C$. Hence, a vertex in $V(G)\setminus V(C)$ is contained in $N_G(V(C))$  if and only if it is contained in $N_{G^*}(V(C^*))$. This implies that $E_G(V(C), V_i)\neq \emptyset$ if and only if $E_{G^*}(V(C^*), V_i) \neq \emptyset$ for each $i\in \{1, \ldots, m\}$. As we have assumed $G$ to be connected, it follows by definition of $D_i$ that $E_G(V(C), V_i) \neq \emptyset$ and hence $E_{G^*}(V(C^*), V_i)\neq \emptyset$ for $i\in \{1,\ldots, m\}$. This completes the proof of \ref{itm:seventh}.

\item We still need  to construct $H^{\ast}$, $ A^{\ast}$, $B^{\ast}$, and $\mathfrak{b}^{\ast}: A^{\ast} \rightarrow \mathcal{P}(B^{\ast})$ in such a way that $(G^{\ast}, A^{\ast}, B^{\ast}, \mathfrak{b}^{\ast})$ represents $H^{\ast}$, $G^{\ast}$ is a $k$-basis $G^{\ast}$ and $\mathfrak{b}^{\ast}: A^{\ast} \rightarrow \mathcal{P}(B^{\ast})$ is a $k$-representation and such that additionally, statements \ref{itm:fifth}- \ref{itm:sixth} hold. We start with  $A^\ast$, $B^\ast$, and $\mathfrak{b}^\ast$ as follows (note that by Part \ref{itm:represents} of Definition \ref{def:basis+representation}, this construction suffices in order to implicitly define a graph $H^{\ast}$): 

\begin{itemize}
    \item We begin with the definition of $A^{\ast}$ and set $A^{\ast}:=A$. As by  assumption $\mathfrak{b}: A\rightarrow \mathcal{P}(B)$ is a $k$-representation, we have $A\cap B = \emptyset$. Using $W = V(C)$ and the fact that $C$ is a connected component of $G[B]$, we conclude $A \cap W = \emptyset$. Hence, $A \subseteq V(G)\setminus W\subseteq V(G^{\ast})$. In particular, $A^\ast=A$ then also implies $A^{\ast}\subseteq V(G^{\ast})$.
    \item Next, we define $B^{\ast} := V(G^{\ast}) \setminus A^{\ast}$.
    \item Next, we define $\mathfrak{b}^{\ast}:A^\ast\rightarrow \mathcal{P}(B^\ast)$ as follows: Let $a\in A^{\ast}$ and define 
    \begin{equation}\mathfrak{b}^{\ast}(a) := \left(\mathfrak{b}(a) \setminus W\right) \cup \{g^{-1}(i): i\in g(\mathfrak{b}(a)\cap W)\}. \label{eq:def_repr}
    \end{equation} Note that the definition of $\mathfrak{b}^{\ast}$ is illustrated by Figure \ref{fig:shrinking2}. In other words, for $\mathfrak{b}^\ast(a)$, we take $\mathfrak{b}(a)$ and first remove all vertices of $W$ (as these are not in $W^\ast$), but then add those vertices of $W^\ast$ that correspond to the color class of at least one of the vertices in $\mathfrak{b}(a)\cap W$.
\end{itemize}

\item Note that the above  definition of $\mathfrak{b}^\ast$ immediately implies \ref{itm:fifth}, as $\mathfrak{b}^{\ast}(a) \setminus W^{\ast} = \mathfrak{b}^{\ast}(a) \setminus \{g^{-1}(i): i\in \{1, \ldots, r\}\} = \mathfrak{b}(a) \setminus W.$ 

\item We next want to show \ref{itm:second}, i.e., we want to show that $C^{\ast} = G^{\ast}[W^{\ast}]$ is a connected component of $G^{\ast}[B^{\ast}]$. In order to do so, from $A^{\ast} = A\subseteq V(G)\setminus W = V(G^{\ast}) \setminus W^{\ast}$, where the last equality holds because of \ref{itm:third}, and $B^{\ast} = V(G^{\ast}) \setminus A^{\ast}$, we conclude $W^{\ast} \subseteq B^{\ast}$. So for \ref{itm:second}, note that in order to show that $C^{\ast} = G^{\ast}[W^{\ast}]$ is a connected component of $G^{\ast}[B^{\ast}]$, it suffices to show $N_{G^{\ast}}(W^{\ast}) \subseteq A^{\ast}$. Also note that the Zykov operations used in the construction of $G^{\ast}$ ensure that $N_G(W) = N_{G^{\ast}}(W^{\ast})$. So from $N_G(W) \subseteq A$ (as $C$ is a connected component of $G[B]$), $N_G(W) = N_{G^{\ast}}(W^{\ast})$ and $A^{\ast} = A$ together we conclude $N_{G^{\ast}}(W^{\ast}) \subseteq A^{\ast}$. Thus,  \ref{itm:second} holds.

\item By Part \ref{itm:represents} of Definition \ref{def:basis+representation} the quadruple $(G^*, A^*, B^*, \mathfrak{b}^*)$ defines a graph $H^{\ast}$. Before we can show \ref{itm:sixth} we have to show that $\mathfrak{b}^*$ is indeed a $k$-representation of $H^*$.

\begin{enumerate} \item At first we check Condition \ref{itm:condition1} of Definition \ref{def:basis+representation}, i.e., we check if $A^{\ast} \cap V(Z) = \emptyset$ for each non-trivial block $Z$ in $G^{\ast}$. So let $Z$ be such a block. We now consider several cases.

\begin{itemize} \item If $V(Z) \cap W^\ast = \emptyset$, then we conclude with the help of \ref{itm:third} that $Z$ is a subgraph of $G$. In this case $Z$ is contained in a non-trivial block of $G$. By assumption Condition \ref{itm:condition1} is true for $(G, A, B, \mathfrak{b})$, implying $A \cap V(Z) = \emptyset$. This, together with $A^\ast = A$, implies $A^\ast \cap V(Z) = \emptyset$ as desired. 

\item Next, we consider the case that $V(Z) \subseteq W^{\ast}$. Note that this implies that $Z$ is a subgraph of $C^\ast$. Then we conclude from \ref{itm:second} that $V(Z) \subseteq B^{\ast}$, which again implies $A^\ast \cap V(Z) = \emptyset$ as desired.

 \item Last, we consider the case in which $V(Z) \cap W^\ast \neq \emptyset$ and $V(Z) \not\subseteq W^\ast$. Our aim is to show that this case cannot happen. Anyhow, in this case, $Z$ contains at least one vertex of $W^\ast$ and at least one vertex that is not contained in $W^\ast$. We now argue that, as $Z$ is a non-trivial block and as $V(Z)\not\subseteq W^{\ast}$, this must imply that there must be at least two edges $e_1, e_2$ $\in E(Z)$ which are elements in $E_{G^\ast}(W^\ast, V(G^\ast)\setminus W^\ast)$ and which even belong to a cycle $D$ of $Z$. To see this, first note that $E(Z) \cap E_{G^\ast}(W^\ast, V(G^\ast)\setminus W^\ast)$ cannot be empty as we have at least one vertex $z\in V(Z)\cap W^\ast\subseteq V(Z)$ and one vertex $z' \in V(G^\ast)\setminus W^\ast\subseteq V(Z)$, so (as $Z$ is connected) there must be a path from $z$ to $z'$ containing an edge from $E(Z) \cap E_{G^\ast}(W^\ast, V(G^\ast)\setminus W^\ast)$. Now assume that $E(Z) \cap E_{G^\ast}(W^\ast, V(G^\ast)\setminus W^\ast)$ contains only one element, say $e_1$. As $Z$ is non-trivial, $E(Z)$ contains at least one more edge $f\neq e_1$ as well as a cycle $D$ containing both $e_1$ and $f$. So this cycle $D$ contains at least one element of $W^{\ast}$ and at least one element of $V(G^{\ast})\setminus W^{\ast}$ (as $e_1\in E(D)$), which implies that in fact it must contain \emph{two} edges connecting $W^{\ast}$ and $V(G^{\ast})\setminus W^{\ast}$, so there must be another such edge $e_2$. Thus, the assumption was wrong, which is why we may conclude $|E(D) \cap  E_{G^\ast}(W^\ast, V(G^\ast)\setminus W^\ast)|\geq 2$ as desired. 
 We may assume that $e_i = \{v_i, w_i\}$ with $v_i\in W^\ast, w_i \notin W^\ast$ for $i=1,2$ (note that $v_1=v_2$ is possible). Furthermore, we may assume that there is a segment $D'$ of $D$ such that $w_1, w_2$ are the endpoints of $D'$ and $D'$ is a subgraph of $G^\ast[V(G^\ast) \setminus W^\ast]$ (note that this segment is -- of the two paths connecting $w_1$ and $w_2$ in $D$ -- the path \emph{not} containing $v_1$ and $v_2$). By \ref{itm:third}, we conclude that $D'$ is contained in $G[V(G) \setminus W]$. As $v_1$ and $v_2$ are obtained by the Zykov operations of contracting non-adjacent vertices from $W$, we conclude that there are some $v'_1, v'_2 \in V(C)$ with $e'_i = \{v'_i, w_i\}\in E_G(W, V(G) \setminus W)$ for $i=1,2$ (again, note that $v_1'=v_2'$ is possible). 

However, by assumption  $C$ is a connected component of $G[B]$, from which several conclusions can be drawn. First, as $C$ is connected, there is a $v'_1$-$v'_2$-path $P$ in $C$, from which -- together with $e'_1$, $e'_2$ and $D'$ -- we obtain a cycle $D''$ in $G$ which contains $e'_1$ and  $e'_2$. On the other hand, using Lemma \ref{lem:bridges_in_basis}, we conclude that $e'_1, e'_2$ are bridges in $G$, which is a contradiction, because $e'_1, e'_2$ are contained in a cycle in $G$. This contradiction shows that the last case cannot happen. 
\end{itemize}

So we conclude that in all cases that can actually happen, we have $A^\ast \cap V(Z)=\emptyset$, which shows that Condition \ref{itm:condition1} is  true for $(G^\ast, A^\ast, B^\ast, \mathfrak{b}^\ast)$.

\item In order to check the remaining Conditions \ref{itm:condition2} and \ref{itm:condition3} of Definition \ref{def:basis+representation} for $(G^\ast, A^\ast, B^\ast, \mathfrak{b}^{\ast})$, we have to show that we have $\deg_{G^{\ast}}(a) \geq 2$ and $\deg_{G^{\ast}}(a) + |\mathfrak{b}^{\ast}(a)| \leq k$ for all $a \in A^\ast$.

By construction of $G^{\ast}$, it is true that $\deg_{G^{\ast}}(x) = \deg_G(x)$ for each $x\in V(G^{\ast}) \setminus W^{\ast} = V(G) \setminus W$, which is simply a consequence of Lemma \ref{lem:neighbor_of_component}. In particular, this implies that  $\deg_{G^{\ast}}(a) = \deg_G(a)$ is true for $a\in A^{\ast}$, because  we know from \ref{itm:second} that $A^{\ast} \subseteq V(G^{\ast}) \setminus W^{\ast}$. By Definition \ref{def:basis+representation} and our assumption that $\mathfrak{b}$ is a $k$-representation, we know that $\deg_G(a) \geq 2$ for $a\in A$. Additionally, by  \ref{itm:fifth} we know $A = A^{\ast}$. Together these facts imply $\deg_{G^{\ast}}(a) \geq 2$ for $a\in A^{\ast}$, which is the first inequality.

Now we show the second inequality. By Condition \ref{itm:condition3} for $(G, A, B, \mathfrak{b})$ we know that $\deg_G(a) + |\mathfrak{b}(a)| \leq k$ for every $a\in A$. Furthermore, by \ref{itm:fifth}, we know that $\mathfrak{b}(a) \setminus W = \mathfrak{b}^{\ast}(a) \setminus W^{\ast}$. As we have already shown  that $\deg_{G^{\ast}}(a) = \deg_G(a)$, it only remains to show that $|\mathfrak{b}^{\ast}(a)| \leq |\mathfrak{b}(a)|$, and thus, by $|\mathfrak{b}^{\ast}(a)|= |\mathfrak{b}^{\ast}(a) \setminus W^\ast|+|\mathfrak{b}^{\ast}(a) \cap W^\ast|$ and $|\mathfrak{b}(a)|= |\mathfrak{b}(a) \setminus W|+|\mathfrak{b}(a) \cap W|$, we only need to show $|\mathfrak{b}^{\ast}(a) \cap W^\ast|\leq |\mathfrak{b}(a) \cap W|$. Using the definition of $\mathfrak{b}^{\ast}$, we conclude this inequality from $\mathfrak{b}^{\ast}(a)\cap W^{\ast} = \{g^{-1}(i)\in W^{\ast}: i\in g(\mathfrak{b}(a)\cap W)\} = \{g^{-1}(g(v))\in W^{\ast}: v \in \mathfrak{b}(a)\cap W\}$ for $a\in A$. 
\end{enumerate}
In summary, we can conclude that $\mathfrak{b}^{\ast}$ is indeed a $k$-representation. This, in turn, finally allows us to define $H^\ast$ such that Conditions \ref{itm:first} - \ref{itm:second} are fulfilled as shown above: Simply let  $H^{\ast}$ be the graph defined by $(G^{\ast}, A^{\ast}, B^{\ast}, \mathfrak{b}^{\ast})$, i.e. $V(H^{\ast})= V(G^{\ast})$ and $E(H^{\ast})= E(G^{\ast}) \cup \{(a, b): a\in A^{\ast}, b\in \mathfrak{b}^{\ast}(a)\}$. 

\item It remains to show \ref{itm:sixth}, i.e., we need to prove that $\chi(H^{\ast}) \geq \chi(H)$. Let $l = \chi(H^{\ast})$ and let $f^{\ast}: V(G^{\ast}) \rightarrow \{1, \ldots, l\}$ be some $l$-coloring of $H^{\ast}$. Our aim now is to construct a function $f: V(G) \rightarrow \{1, \ldots, l\}$ which is an $l$-coloring of $H$, which will prove the desired result. 

We construct $f: V(G) \rightarrow \{1, \ldots, l\}$ in the following way:

\begin{equation*} 
    f(v)=\begin{cases} f^{\ast}(v) & \text{if $v\in V(G) \setminus V(C)$},\\ f^{\ast}(g^{-1}(i))  &\text{if $v\in g^{-1}(i)$.}\end{cases}
\end{equation*}

It remains to show that $f$ is an $l$-coloring of $H$. For $e = \{v,w\} \in E(H)$ we have to show $f(v) \neq f(w)$. We use a case distinction:

\begin{enumerate}
    \item We first consider the case $e \subseteq V(G) \setminus W$. By \ref{itm:third} and \ref{itm:fifth} we know $H[V(G)\setminus W] = H^{\ast}[V(G^{\ast})\setminus W^{\ast}]$. This implies $e\in E(H^{\ast})$. As $f^{\ast}$ is a coloring of $H^{\ast}$, we have $f^{\ast}(v) \neq f^{\ast}(w)$. Then $f(v) \neq f(w)$ follows, because $f(v) = f^{\ast}(v)$ and $f(w) = f^{\ast}(w)$ by definition of $f$.

    \item Next, we consider the case that $e\in E(G)$ and $v\in W$ and $w\in N_G(W)$. By construction of $G^{\ast}$, we know for all $i\in \{1, \ldots, r\}$ that $\{w,g^{-1}(i)\}\in E(G^{\ast})$ if and only if there is some $w'\in g^{-1}(i)$ such that $\{w,w'\} \in E(G)$. Hence the existence of $\{v,w\}$ in $E(G)$ implies $\{g^{-1}(j), w\}\in E(G^{\ast})$ with $j = g(v)$. Now $f^{\ast}(w) \neq f^{\ast}(g^{-1}(j))$ (as $f^{\ast}$ is a coloring) together with $f(w) = f^{\ast}(w)$ and $f(v) = f^{\ast}(g^{-1}(j))$ imply $f(v) \neq f(w)$. Note that the case in which $e\in E(G)$ and $w\in W$ and $v\in N_G(W)$ follows analogously.

    \item Now, consider the case that $e\in E(G)$ and $e\subseteq V(C)=W$. In this case, we have $f(v) = f^{\ast}(g^{-1}(i))$ and $f(w) = f^{\ast}(g^{-1}(j))$ with $i=g(v)$ and $j=g(w)$. We note that $i\neq j$ as $g$ is an $r$-coloring of $C$. By \ref{itm:third} we know that $G^{\ast}[W^{\ast}]$ is a clique. So $\{g^{-1}(i), g^{-1}(j)\}\in E(G^{\ast})$ and $f^{\ast}(g^{-1}(i)) \neq f^{\ast}(g^{-1}(j))$, from which $f(v) \neq f(w)$ follows.

    \item The last case we need to consider is the case in which $e\in E(H) \setminus E(G)$. As $H$ is represented by $(G, A, B, \mathfrak{b})$, this means that $e=\{v,w\}$ for some $v\in A$ and $w\in \mathfrak{b}(v)$. As $C$ is a connected component of $G[B]$, we conclude  $A\subseteq V(G)\setminus W$ and thus $v\in V(G)\setminus W$. Note that we may assume without loss of generality that we are not in the first case (as this has already been covered), i.e., we may assume that $w\in W$. Now let $w^{\ast} = g^{-1}(g(w))\in W^{\ast}$. Then $w^{\ast} \in \mathfrak{b}^{\ast}(v)$ by definition of $\mathfrak{b}^{\ast}$. This implies that $\{v,w^{\ast}\}$ is represented by $\mathfrak{b}^{\ast}$ and thus $\{v,w^{\ast}\} \in E(H^{\ast})$. Hence $f^{\ast}(v) \neq f^{\ast}(w^{\ast})$. Using the definition of $f$ we conclude $f(v) = f^{\ast}(v)$ and $ f(w) = f^{\ast}(w^{\ast})$, and hence $f(v) \neq f(w)$ as desired.
\end{enumerate}
Thus, $f$ is an $l$-coloring of $H$, which shows that we indeed have $\chi(H)\leq \chi(H^\ast)$. \end{itemize}

This shows that our construction of $G^{\ast}, H^{\ast}$ and $\mathfrak{b}: A^{\ast} \rightarrow \mathcal{P}(B^{\ast})$ is such that the statements \ref{itm:first}-\ref{itm:sixth} all hold and thus completes the proof.
\end{proof}

\begin{remark}
    In our proof of Lemma \ref{lem:shrinking} we used only Zykov operations of the second type, i.e.,  identifications of non-adjacent vertices. Alternatively, in the first part of the proof a slightly different construction could be used which has the advantage that it illuminates how the reduction which is described in Lemma \ref{lem:shrinking} is basically an repeated application of \emph{both} Zykov operations described by  Definition \ref{def:zykov}. We constructed $C^*$ with the help of $g$ which is a \emph{minimal} coloring of $G[W]$. Another option would be to proceed without this  minimality assumption and to suppose instead that $g$ needs at most $k$ colors. In this case, the identification of vertices in the same color class would not necessarily lead to a set $W^*$ forming a clique. However, in this case, we could use Zykov operations of the first kind and add additional edges between all pairs of non-adjacent vertices from $W^*$ to ensure that $W^*$ is a clique. By the first part of Lemma \ref{lem:zykov2}, it would be guaranteed that the result of these operations can still be used to derive an upper bound for the chromatic number of $G$. Hence, this alternative construction leads to an alternative (yet slightly more extensive) proof of the statement employing both Zykov operations.    
\end{remark}

Next, we will use Lemma \ref{lem:shrinking} iteratively in order to derive the following result.

\begin{proposition}\label{prop:shrinking}
Let $k \geq 2$ and let $G$ be a connected graph with $\chi(G) \leq k$. Let $H$ be a graph with $k$-basis $G$ and $k$-representation $\mathfrak{b}: A\rightarrow \mathcal{P}(B)$. Then there exist $H'$, $G'$, $A'$, $B'$ and $\mathfrak{b}'$ such that $H'$ is a graph with $k$-basis $G'$ and $k$-representation $\mathfrak{b}': A' \rightarrow \mathcal{P}(B')$ having the following properties:
\begin{itemize}[label=-]
\item $A' =A$,
\item $\chi(H) \leq \chi(H')$.
\item Every connected component of $G'[B']$ is a clique with at most $k$ vertices.
\item If $G$ is connected, then $G'$ is also connected.
\end{itemize}
\end{proposition}
\begin{proof} 
 Let $m$ denote the number of connected components of $G[B]$ and let $C_1, \ldots, C_m$ denote the corresponding connected components. The proof idea now is to construct sequences $G_0, G_1, \ldots, G_m$ and $H_0, \ldots, H_m$ with $G_0 = G$ and $H_0 = H$ such that for each $i\in \{0, \ldots, m\}$ there is a $k$-representation $\mathfrak{b}_i: A_i \rightarrow \mathcal{P}(B_i)$ representing $H_i$ with $k$-basis $G_i$, and such that this sequence has the following property: For each $i\in \{1, \ldots, m\}$, the graph $G_i[B_i]$ will have $C'_1, \ldots, C'_i, C_{i+1}, \ldots, C_m$ as connected components, where $C'_i$ is a clique of size at most $k$. This way, in particular, $G_m[B_m]$ has the connected components $C'_1, \ldots, C'_m$. Moreover, the sequence we construct will have the property that $\chi(H) \leq \chi(H_1) \leq \chi(H_2) \leq \ldots \leq \chi(H_m)$.  Finally, by setting $G'=G_m$, $H'=H_m$, $\mathfrak{b}'=\mathfrak{b}_m$, $A'=A_m$ and $B'=B_m$, we will derive the desired result.  

Now we start our construction by setting  $G_0 = G$, $H_0 = H$, $A_0 = A$, $B_0 = B$, and $\mathfrak{b}_0 = \mathfrak{b}$. For $i=0,\ldots,m-1$ we repeat the following: Each time $G_i$, $H_i$ and $\mathfrak{b}_i: A_i \rightarrow \mathcal{P}(B_i)$ are constructed  for some $i<m$, we apply Lemma \ref{lem:shrinking} to $G_i$, $H_i$, $\mathfrak{b}_i$ and the connected component $C_{i+1}$ to derive $G_{i+1}$, $H_{i+1}$ and $\mathfrak{b}_{i+1}: A_{i+1} \rightarrow \mathcal{P}(B_{i+1})$ as the result of this construction. By Property \ref{itm:first} of Lemma \ref{lem:shrinking} we know that $V(G_{i+1}) = (V(G_i) \setminus V(C_{i+1})) \cup V(C'_{i+1})$, where -- using Property \ref{itm:second} -- we may assume that $C'_{i+1}$ is a connected component of $G_{i+1}[B_{i+1}]$. From Property \ref{itm:first} we also know  $V(C_{i+1}) \cap V(C'_{i+1}) = \emptyset$, so it is clear then that $C_{i+1}$ is not a connected component of $G_{i+1}[B_{i+1}]$. 

To show that our construction has the desired properties we have to show that all other connected components of $G_i[B_i]$ are not changed by the construction of $G_{i+1}[B_{i+1}]$. By Property \ref{itm:third} of Lemma \ref{lem:shrinking} we know $G_i[V(G_i) \setminus V(C_{i+1})] = G_{i+1}[V(G_{i+1}) \setminus V(C'_{i+1})]$. As $C_{i+1}$ is a connected component of $G_i[B_i]$ and $C'_{i+1}$ is a connected component of $G_{i+1}[B_{i+1}]$, we conclude from this that  $G_{i+1}[B_{i+1} \setminus V(C'_{i+1})] = G_i[B_i \setminus V(C_{i+1})]$. Here $G_i[B_i \setminus V(C_{i+1})]$ is the union of all connected components of $G_i[B_i]$ different from $C_{i+1}$ and $G_{i+1}[B_{i+1} \setminus V(C'_{i+1})]$ is the union of all connected components of $G_{i+1}[B_{i+1}]$ different from $C'_{i+1}$. However, if both graphs are identical, their connected components are identical, too.

With the sequence $H_1, \ldots, H_m$ constructed as described above, we can apply Property \ref{itm:sixth} of Lemma \ref{lem:shrinking} to conclude $\chi(H) \leq \chi(H_1) \leq \chi(H_2) \leq \ldots \leq \chi(H_m)$. We note that in each construction step we have replaced the connected component $C_i$ with some $C'_i$ which is by Property \ref{itm:fourth} of Lemma \ref{lem:shrinking} a clique with at most $k$ vertices. Furthermore,  from Property \ref{itm:fifth} we conclude $A = A_0 = A_1 = \ldots = A_m$. Note that by Property \ref{itm:seventh} if $G$ is connected, then all elements in the series $G_0, \ldots, G_m$ are connected, too. Thus, with $G' = G_m$, $H' = H_m$, $\mathfrak{b}' = \mathfrak{b}_m$, $A' = A_m$ and $B' = B_m$ we get the desired result, which concludes the proof. 
\end{proof}

The preceding proposition allows us to simplify the task of proving Theorem \ref{thm:main}. 
If $H$ is a graph with $k$-basis $G$ and $\chi(G) \leq k$, then we may, without loss of generality, 
assume that $(G, A, B, \mathfrak{b})$ represents $H$ in such a way that each connected component of $G[B]$ 
is a clique of size at most $k$. The usefulness of this assumption for the proof of the main result 
will become apparent in the following lemma. In this lemma certain assumptions are made, among them 
that each connected component of $G[B]$ contains at most $k$ vertices and that at least one connected 
component of $G[B]$ is a clique. These assumptions are justified if, as guaranteed by Proposition 
\ref{prop:shrinking}, each connected component of $G[B]$ is in fact a clique with at most $k$ vertices.

\begin{lemma}\label{lem:modify1}
 Let $k \geq 2$ and let $H$ be a graph with $k$-basis $G$ such that $\chi(G) \leq k$ and with $k$-representation $\mathfrak{b}: A\rightarrow \mathcal{P}(B)$. Assume that each connected component of $G[B]$ has at most $k$ elements and that $C$ is a connected component of $G[B]$ which is a clique. Let $v\in N_G(V(C))$ and assume $\mathfrak{b}(v) \subseteq V(C)$. Then, there is a graph $H^{-}$ with $k$-basis $G^{-}$ and $k$-representation $\mathfrak{b}^{-}: A^{-} \rightarrow \mathcal{P}(B^{-})$ with the following properties:
\begin{enumerate}[label=(\roman*)]
\item \label{itm:modify1_1} $|A^{-}| < |A|$.
\item \label{itm:modify1_2} Each block in $G^-$ contains at most $k$ elements.
\item \label{itm:modify1_3} $\chi(H) \leq \chi(H^{-})$.
\end{enumerate}
\end{lemma}

\begin{proof}
We define $p:=|V(C)|\leq k$ and prove the claims of the lemma in separate steps. The first step is to find a $k$-basis $G^{-}$ and a $k$-representation $\mathfrak{b}^{-}: A^{-} \rightarrow B^{-}$ and then to show Properties \ref{itm:modify1_1}, \ref{itm:modify1_2} and \ref{itm:modify1_3} as stated in the lemma. We make a case distinction based on $p$.

\begin{enumerate}
\item First, consider the case $p<k$. In this case, we construct $G^-$ by setting $V(G^-) = V(G)$ and $E(G^-) = E(G) \cup \{\{v,w\}: w\in \mathfrak{b}(v)\}$. Moreover, we set $A^{-} = A\setminus \{v\}$, $B^{-}= B\cup \{v\}$ and  $\mathfrak{b}^{-} = \mathfrak{b}|_ {A^{-}}$. From these definitions it is clear that $\deg_{G^-}(x) \geq \deg_G(x)$ for all $x\in V(G)$. Together with $A^- \subseteq A$ and $\deg_G(a) \geq 2$ (by Condition \ref{itm:condition2} for $(G, A, B, \mathfrak{b})$) for all $a\in A$ we conclude that $\deg_{G^-}(a)\geq 2$ for all $a\in A^-$ as well, which proves Condition \ref{itm:condition2} for $(G^-, A^-, B^-, \mathfrak{b}^-)$. 

Furthermore, it can easily be seen that if $\deg_{G^-}(x) > \deg_G(x)$, then $x\in \{v\}\cup \mathfrak{b}(v)$ (this is a direct consequence of the definition of $E(G^-)$). However, our definition of $B^-$ implies $\{v\}\cup \mathfrak{b}(v) \subseteq B^-$, so we necessarily have $\deg_{G^-}(a) = \deg_G(a)$ for all $a\in A^-$. Moreover, the definition of $\mathfrak{b}^-$ implies $\mathfrak{b}^-(a) = \mathfrak{b}(a)$ for each $a\in A^-$. Together, both equations imply Condition \ref{itm:condition3} as $\deg_{G^-}(a)+|\mathfrak{b}^{-}(a)|=\deg_{G}(a)+|\mathfrak{b}(a)|\leq k$ for all $a\in A^-$, where the latter inequality stems from the fact that $(G,A,B,\mathfrak{b})$ represents $H$ and $G$ is a $k$-basis of $H$. 

It remains to check Condition \ref{itm:condition1}. Let $Z$ be a non-trivial block in $G^-$. We have to show that $A^-\cap V(Z) = \emptyset$. If $v\notin V(Z)$, then $Z$ is contained in a non-trivial block in $G$ (as $G-v = G^--v$) and $A^- \subseteq A$ implies $A^-\cap V(Z) \subseteq A\cap V(Z) = \emptyset$, which completes the case $v\notin V(Z)$. 

Now we assume $v\in V(Z)$. Let $W = V(C) \cup \{v\}$. We now consider two subcases. First, assume $V(Z)\subseteq W$. 
 As $C$ is a connected component of $G[B]$ and by the definition of $B^-$, we conclude $W\subseteq B^-$. So as $V(Z)\subseteq W$, we clearly have $V(Z) \cap A^- = \emptyset$, which completes the proof of this case.
 
 So let us consider the remaining case in which we have $v \in V(Z)$ and thus $v \in V(Z)\cap W$ as well as $V(Z)\not\subseteq W$. We will show that this situation leads to a contradiction. So there exists some $w\in V(Z) \setminus W$. Then, as $Z$ is a non-trivial block with $v,w\in V(Z)$, there exists a cycle $D$ with $v,w\in V(D)$. Moreover, note that $D$ contains at least two different edges $e_1, e_2 \in E_{G^-}(W, V(G^-) \setminus W)$. This is as above due to the fact that $v\in W$ and $w\not\in W$ and $v,w\in V(D)$.  
 
 We now claim that $e_1, e_2\in E(G)$. Otherwise,  if for example $e_1\in E(G^-) \setminus E(G)$, then the definition of $G^-$ implies that $e_1 = \{v,w\}$ with $w\in \mathfrak{b}(v)$. But we have assumed $\mathfrak{b}(v)\subseteq V(C)$, which then implies $w\in \mathfrak{b}(v)\subseteq V(C)\subseteq W$, contradicting $w \not\in W$. So this shows that the assumption was wrong and $e_1 \in E(G)$ (and the same holds analogously for $e_2$).  
 
 Next, we claim that both $e_1$ and $e_2$ are bridges in $G$. Since exactly one endpoint of each of these edges lies in $W = V(C) \cup \{v\}$, for each $e_i$ ($i = 1,2$), we have two possibilities: Either $e_i \in E_G(V(C), V(G)\setminus V(C))$ or $v\in e_i$. Using Lemma \ref{lem:bridges_in_basis}, we then conclude that $e_i$ is a bridge in both cases. 
 
 Let $e_i = \{x_i, y_i\}$ such that $x_i\in W$ and $y_i \notin W$. Let $D'$ be the segment of $D$ with endpoints $y_1, y_2$ such that $x_1, x_2\notin V(D')$. It may happen that $|E_{G^-}(W, V(G^-) \setminus W)| > 2$, but if we choose $e_1, e_2$ carefully we can nevertheless assume that $D'$ is contained in $G^-[V(G^-)\setminus W]$. As $D'$ is contained in $D-v$, it is also contained in $G-v$ (as $G^--v = G-v)$. On the other hand, $G[W]$ is connected and therefore contains an $x_1$-$x_2$-path $P$. Let $D''$ be the cycle composed of $D', e_1, e_2, P$. Then,  $D''$ is a cycle in $G$ which contains $e_1, e_2$, which contradicts the fact that $e_1, e_2$ are bridges as shown above. This contradiction shows that we must indeed have $V(Z) \subseteq W$ and thus completes the proof of Condition \ref{itm:condition1}.

Now that we have shown all three conditions of Definition \ref{def:basis+representation} for $(G^-, A^-, B^-, \mathfrak{b}^-)$, we finally want to show the remaining assertions of  Lemma \ref{lem:modify1}.

The definition of $A^-=A\setminus \{v\}$ implies $|A^-| < |A|$, which is \ref{itm:modify1_1}. In order to show \ref{itm:modify1_2}, let $Z$ be a non-trivial block of $G^-$. In case $v\notin V(Z)$, we know that $Z$ is contained in a non-trivial block in $G$ (as $G^--v = G-v$). By Condition \ref{itm:condition1} for $(G, A, B, \mathfrak{b})$, we conclude that $Z$ is contained in $G[B]$. As $Z$ is a connected subgraph of $G[B]$, it is contained in a connected component of $G[B]$.  We have assumed that each connected component in $G[B]$ contains at most $k$ vertices, so we conclude that \ref{itm:modify1_2} holds in this case. If, however, we have $v\in V(Z)$, then, as shown above (in the course of the proof of Condition \ref{itm:condition1}), we have $V(Z) \subseteq W$. However, $|W| = |V(C)| + 1 = p+1 \leq k$, because we have assumed $p<k$. Hence \ref{itm:modify1_2} holds in this case, too. 

Finally, we define $H^-$ to be the graph defined by $(G^-, A^-, B^-, \mathfrak{b}^-)$. In order to show $\chi(H)\leq \chi(H')$, we first show that $E(H) = E(H^-)$: 
\begin{align*}
    E(H^-) &= E(G^-) \cup \{\{x,y\}: x\in A^-, y\in \mathfrak{b}^-(x)\} \\
    &= ( E(G) \cup \{\{v,w\}: w\in \mathfrak{b}(v)\})  \cup \{\{x,y\}: x\in A^-, y\in \mathfrak{b}^-(x)\}\\
    &= E(G) \cup \{\{x,y\}: x\in A^-\cup\{v\}, y\in \mathfrak{b}(x)\}\\
    &= E(G) \cup \{\{x,y\}: x\in A, y\in \mathfrak{b}(x)\}\\
    &= E(H). 
\end{align*}
Obviously, this implies $H = H^-$ and thus $\chi(H) = \chi(H^-)$, which implies \ref{itm:modify1_3} as desired. This completes the proof of Lemma \ref{lem:modify1} for the case $p<k$.

\item Next, before we can consider the case $p=k$, we need to prove the following inequality, which we will use subsequently to prove the desired assertions:
\begin{equation}|N_H(v)\cap V(C)| \leq k-1 .\label{eq:claim}\end{equation}

Note that by Condition \ref{itm:condition2} of Definition \ref{def:basis+representation}, we have $\deg_G(v) \geq 2$. Moreover, we also have $\deg_H(v) \leq k$ and thus $|N_H(v)|\leq k$ by Condition \ref{itm:condition3}. Now, the first inequality implies that there are at least two vertices adjacent to $v$ in $G$. Therefore, by Lemma \ref{lem:neighbor_of_component}, at least one vertex adjacent to $v$ in $G$ is \emph{not} in $C$. Then, due to the second inequality, at most $k-1$ vertices in $V(C)$ are adjacent to $v$ in $H$. This implies \eqref{eq:claim}. 

\item Now, we consider the case $p=k$. Due to Inequality \eqref{eq:claim}, there is some $w\in V(C)$ such that $v$ and $w$ are \emph{not} adjacent. Hence we can construct $G^-$ in this case by applying the Zykov operation of identifying $v$ and $w$ to $G$ (cf. Figure \ref{fig:modify1} for an example). We subsequently denote by $w'$ the new vertex of $G^-$, which results from identifying $v$ and $w$.

Next we construct $\mathfrak{b}^{-}: A^{-} \rightarrow \mathcal{P}(B^{-})$.  Set $A^{-}:= A\setminus \{v\}, B^{-}:= V(G^{-}) \setminus A^{-}$. For the definition of $\mathfrak{b}^{-}$ consider some $a\in A^{-}$. If $w\notin \mathfrak{b}(a)$, then set $\mathfrak{b}^{-}(a) := \mathfrak{b}(a)$. Otherwise set $\mathfrak{b}^{-}(a) := (\mathfrak{b}(a) \setminus \{w\})\cup \{w^{\prime}\}$. Note that in the following, to simplify the notation, we will denote the set $V(C) \cup \{v\}$ by $W$ and $(V(C)\setminus \{w\}) \cup \{w'\}$ by $W'$.
 
We need to check that $(G^-, A^-, B^-, \mathfrak{b}^{-})$ fulfills the conditions  of Definition \ref{def:basis+representation}, Part \ref{itm:basis_first}. First, we consider Condition \ref{itm:condition1}. Let $Z$ be a non-trivial block in $G^-$. We consider different cases. If $w'\notin V(Z)$, then $Z$ is contained in a non-trivial block in $G$ (as $G^--w' = G-v-w$). In this case, Condition \ref{itm:condition1} for $(G, A, B, \mathfrak{b})$ implies $V(Z) \cap A = \emptyset$ and thus also $V(Z) \cap A^- = \emptyset$ as desired (as $A^- \subseteq A$ by definition of $A^-$). 
 
On the other hand, if $w'\in V(Z)$, we first show that $V(Z) \subseteq W'$. Suppose this is not the case. Then, $Z$ contains a cycle $D$ in $G^-$ such that $D$ contains at least two edges $e_1, e_2\in E_{G^-}(W', V(G^-)\setminus W')$. We can assume that $e_i = \{x_i, y_i\}$ such that $x_i \in W'$ and $y_i \notin W'$ for $i=1,2$. Furthermore, like in the first part of the proof, we can assume that $e_1, e_2$ are chosen such that there is a segment $D'$ of $D$ with endpoints $y_1, y_2$ and $D'$ is contained in $G^-[V(G^-)\setminus W^-]$. Now for each $e_i$ there are two possibilities: Either $e_i$ is in $E(G)$ or it resulted from the identification of $v$ with $w$. In each case $e_i$ corresponds to some edge $e'_i$ in $E_G(W, V(G) \setminus W)$. Using an analogous argument as in the first case of the proof, as $G[V(G)\setminus W] = G^-[V(G^-) \setminus W^-]$ and $G[W]$ is connected, we can conclude the existence of a cycle $D''$ in $G$ containing $e'_1, e'_2$. But by Lemma \ref{lem:bridges_in_basis}, $e'_1$ and $e'_2$ are bridges, which shows that they cannot be contained in a cycle. This contradiction proves that indeed we must have  $V(Z) \subseteq W'$.  This, together with $W' \subseteq B^-$, shows that Condition \ref{itm:condition1} holds in this case, too. 
 
Now let $a\in A^-$. For the remaining conditions of Definition \ref{def:basis+representation} we have to show that the following inequalities are fulfilled:
\begin{itemize}
    \item $\deg_{G^-}(a) \geq 2$. 
    \item $\deg_{G^-}(a) + |\mathfrak{b}^-(a)| \leq k$. 
\end{itemize}

For the first inequality we note that according to the first step of the construction of $G^-$ we have $G[V(G)\setminus W] = G^-[V(G^-)\setminus W']$ and thus $\deg_{G[V(G)\setminus W]}(a) = \deg_{G^-[V(G^-)\setminus W']}(a)$ for $a\in A^- \subseteq V(G) \setminus W$. Now we claim that $|E_G(a, W)| = |E_{G^-}(a, W')|$ for each $a\in A^-$. As $G^-$ is obtained from $G$ by identifying $v$ with $w$, this claim is equivalent to the statement that there is no $a\in A$ which is adjacent in $G$ to both, $v$ and $w$. To see that such an $a$ cannot exist, assume the contrary: if there was such an $a\in A$, then by Lemma \ref{lem:bridges_in_basis} both $\{a,v\}$ and $\{a, w\}$ would be bridges in $G$. As $G[W]$ is connected, there is a $v$-$w$-path in $G[W]$, which together with $v,a,w$ forms a cycle. This, however, is impossible as $\{a,v\}$ and $\{a, w\}$ are bridges. This contradiction shows that indeed we have  $|E_G(a, W)| = |E_{G^-}(a, W')|$. Altogether we therefore conclude $ \deg_{G^-}(a) = \deg_{G^-[V(G^-)\setminus W']}(a) + |E_{G^-}(a, W')| = \deg_{G[V(G)\setminus W]}(a) + |E_G(a, W)| = \deg_G(a)$. In particular, this shows that by $\deg_G(a) \geq 2$, which is Condition \ref{itm:condition2} for $(G, A, B, \mathfrak{b})$, we also have $\deg_{G^-}(a) \geq 2$ as desired.

For the second inequality, recall that the definition of $\mathfrak{b}^-$ is based on a case distinction. If $w\notin \mathfrak{b}(a)$, then we have  $\mathfrak{b}^{-}(a) = \mathfrak{b}(a)$. As we have just seen that $\deg_{G^-}(a) =  \deg_G(a)$, it is clear that $\deg_{G^-}(a) + |\mathfrak{b}^-(a)| = \deg_{G}(a) + |\mathfrak{b}(a)| \leq k$, where the last inequality stems from Condition \ref{itm:condition3} for $(G, A, B, \mathfrak{b})$. Otherwise, if $w\in \mathfrak{b}(a)$, we have $\mathfrak{b}^-(a)=(\mathfrak{b}(a)\setminus\{w\})\cup\{w'\}$, which implies $|\mathfrak{b}^-(x)| = |\mathfrak{b}(x)|$. Again together with $\deg_{G^-}(x) =  \deg_G(x)$, we derive the second inequality.
 
So in summary, so far we have shown that $\mathfrak{b}^-: A^-\rightarrow \mathcal{P}(B^-)$ is a $k$-representation with $k$-basis $G^-$. Let $H^-$ be the graph defined by $(G^-, A^-, B^-, \mathfrak{b}^-)$. So it remains to show the Properties \ref{itm:modify1_1} - \ref{itm:modify1_3} of Lemma \ref{lem:modify1}. Note that our definition of $A^-$ implies that $|A^-| < |A|$, so \ref{itm:modify1_1} is fulfilled. In order to prove \ref{itm:modify1_2}, let $Z$ be a non-trivial block in $G^-$. Then either $w'\notin V(Z)$, in which case $Z$ is contained in a non-trivial block in $G$ (as $G^-- w' = G-w-v$). Using  Condition \ref{itm:condition1} of Definition \ref{def:basis+representation} for $(G, A, B, \mathfrak{b})$, we conclude $V(Z) \subseteq G[B]$. As by assumption of Lemma \ref{lem:modify1} each connected component of $G[B]$ contains at most $k$ elements, \ref{itm:modify1_2} follows easily in this case. Otherwise, if we have $w'\in V(Z)$, we have $V(Z) \subseteq W'$, as shown above (in the part of the proof where we showed that Condition \ref{itm:condition1} holds). As $|W'| = k$, \ref{itm:modify1_2} follows in this case, too.

The last property we need to show is $\chi(H) \leq \chi(H^-)$. Let $f': V(H^-) \rightarrow \{1, \ldots, l\}$ be an $l$-coloring of $H^-$ for some $l\in \mathbb{N}$. Our aim now is to use $f'$ to construct an $l$-coloring of $H$, which will complete the proof (using the special case $l=\chi(H^-)$).

We define $f: V(H)\rightarrow \{1, \ldots, l\}$ in the following way:
\begin{equation*} 
        f(x)=
        \begin{cases} f'(x) & \text{if $x\notin \{v, w\}$},
        \\f'(w')  & \text{if $x\in \{v,w\}$.}
        \end{cases}
\end{equation*}

In order to prove that $f$ is indeed an $l$-coloring, we consider some $e=\{x,y\}\in E(H)$. We need to show that $f(x) \neq f(y)$. We use a case distinction. 
\begin{enumerate}
    \item Suppose $e\in E(G)$. In this case, if $e\cap \{v,w\} = \emptyset$ we conclude $e\in E(G^-)$ (as $G-v-w = G^--w'$). Thus, we have $f(x) = f'(x)$ and $f(y) = f'(y)$. Then $f'(x) \neq f'(y)$, as $f'$ is a coloring. Hence, $f(x) \neq f(y)$.

    The other possible subcase is that $e\cap \{v,w\} \neq \emptyset$. As $v$ and $w$ are non-adjacent, we conclude $e\neq \{v,w\}$. So we may assume without loss of generality that $x\in \{v,w\}$ and $y\notin \{v,w\}$. Since $G^-$ was obtained from $G$ by identifying $v$ and $w$, we conclude $\{w',y\}\in E(G^-)$. Then $f(x) = f'(w')$ and $f(y) = f'(y)$. As $f'$ is a coloring, we have $f'(w') \neq f'(y)$ and thus derive $f(x) \neq f(y)$ as desired. 

    \item Suppose $e\in E(H) \setminus E(G)$. Then $e$ is represented by $\mathfrak{b}$, and we can assume without loss of generality that $x\in A$ and $y\in \mathfrak{b}(x)$. Now we consider two different cases:
    \begin{enumerate}
        \item First, we consider the case $x=v$. Recall that Lemma \ref{lem:modify1} assumes $\mathfrak{b}(v) \subseteq V(C)$. Hence $y\in V(C)$. As $v$ and $w$ are non-adjacent in $H$, we conclude $y\in W\setminus \{w\} \subseteq W'$. Since $C$ is assumed to be a clique, we have $\{y, w\}\in E(G)$, and by identifying $w$ with $v$, we obtain $\{y, w'\}\in E(G')$. Then $f(v) = f'(w')$ and $f(y) = f'(y)$. As $f'$ is a coloring, we know that $f'(w') \neq f'(y)$ and thus conclude that $f(v) \neq f(y)$.

        \item Now, consider the case $x\in A\setminus \{v\}=A^-$. In this case, we have $f(x)=f'(x)$. Our definition of $\mathfrak{b}^-$ together with $y\in \mathfrak{b}(x)$ imply that either $y = w$ or $y\in \mathfrak{b}^-(x)$. First, consider the case that $y\in \mathfrak{b}^-(x)$. Then, we have $\{x,y\}\in E(H^-)$ and thus $f'(x) \neq f'(y)$, because $f'$ is a coloring of $H^-$. We conclude $f(x)\neq f(y)$. On the other hand, if $y = w$, then $f(y)=f'(w')$, and by our definition of $\mathfrak{b}^-$, it follows that $w' \in \mathfrak{b}^-(x)$. In this case $\{x,w'\}\in E(H^-)$, and therefore we conclude that $f'(x) \neq f'(w')$ and thus also $f(x) \neq f(w)$.
    \end{enumerate}
\end{enumerate}

\end{enumerate}

This completes the proof.
\end{proof}

\begin{figure}
    \centering
\includegraphics{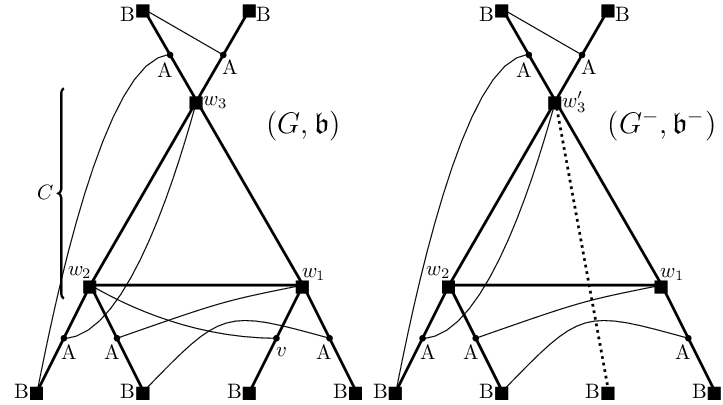}
\caption{This is an example with $k=3$ for the procedure described in the proof of Lemma \ref{lem:modify1}. We assume on the left side that the vertices $w_1$,  $w_2$ and $w_3$ belong to $B$ and constitute a connected component $C$ of $G[B]$ (which is a clique as these vertices form a triangle). Furthermore, we assume that $v$ is in $A$. We see that $v$ is a neighbor of $C$ and $\mathfrak{b}(v) \subseteq V(C)$. In this case, $(G^-, A^-, B^-, \mathfrak{b}^-)$ is obtained by identifying $v$ and $w_3$. In the figure, thick lines indicate which edges belong to the respective basis, and thin lines indicate which edges are represented by the respective function. The dotted line on the right side indicates an edge by which some vertex in $B$ (which is connected with $v$ in the left figure)  ends up as a neighbor of $w'_3$ in the right figure after identifying it with $v$.}
\label{fig:modify1}
\end{figure}

\begin{figure}
    \centering
\includegraphics{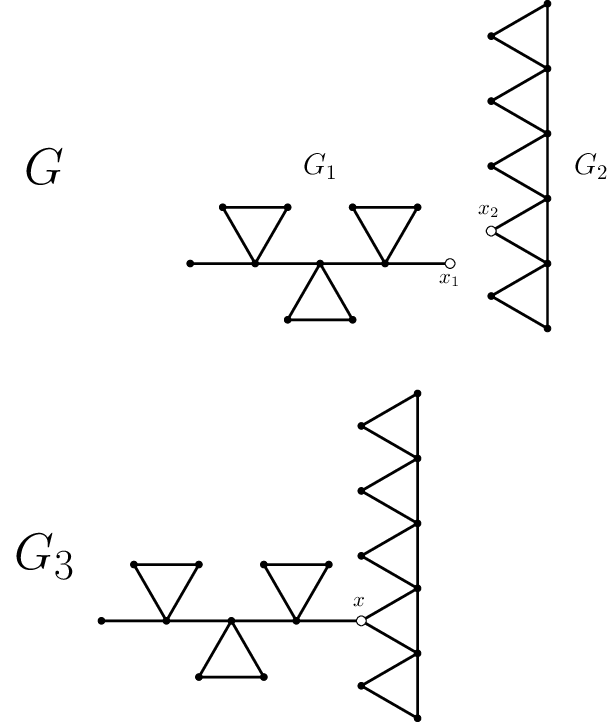}

\caption{This is an example for the Zykov operation described in Definition \ref{def:zykov} which is used in Lemma \ref{lem:construct_basis_by_identifying}. Here, $G_1$ and $G_2$ are the two connected components of graph $G$. By identifying $x_1$ and $x_2$, a new graph $G_3$ is obtained, which has exactly one new vertex that is not contained in $G$, namely $x$.}
\label{fig:identify}
\end{figure}

There is another construction which will be needed for the proof of Theorem \ref{thm:main}. This construction will subsequently be described in Lemma \ref{lem:modify2}, and it requires the following lemma to help us with the construction of a new $k$-basis.

\begin{lemma}\label{lem:construct_basis_by_identifying}
    Let $k\geq 2$ and let $G$ be a graph consisting of two connected components $G_1$ and $G_2$. Let each block in $G$ contain at most $k$ vertices. Let $x_1 \in V(G_1)$ and $x_2 \in V(G_2)$. Let $G_3$ be the graph resulting from identifying $x_1$ and $x_2$ to form a new vertex $x$. Then each block in $G_3$ has at most $k$ vertices. 
\end{lemma}

Before we prove the lemma, note that Figure \ref{fig:identify} illustrates the construction of graph $G_3$.

\begin{proof}[Proof of Lemma \ref{lem:construct_basis_by_identifying}]

    We begin the proof by noting that from the construction of $G_3$ it directly follows that $G'_i = G_3[V(G_i)\setminus \{x_i\}\cup x]$ is isomorphic to $G_i$ for $i=1, 2$, and that $V(G'_1) \cap V(G'_2) = \{x\}$. Let $Z$ be a block in $G_3$. Without loss of generality we can assume that $Z$ is non-trivial (because otherwise $|Z| \leq 2$, in which case there is nothing to show). We have to show that $Z$ contains at most $k$ vertices. Our proof strategy now is to show that $Z$ is contained either in $G'_1$ or in $G'_2$ and thus must be isomorphic to a block contained in $G$. As all these blocks by assumption of the lemma contain at most $k$ vertices, this will complete the proof.
    
    Thus, it only remains to show that $Z$ is contained either in $G'_1$ or in $G'_2$.  Assume this is not the case. Then, there is a cycle $C$ contained in $Z$ such that $C$ contains at least one element of $V(G'_1) \setminus V(G'_2)$ and at least one element of $V(G'_2) \setminus V(G'_1)$. As $V(G'_1) \cap V(G'_2)$ contains only one element (namely $x$), we conclude that $C$ contains at least one edge connecting $V(G'_1) \setminus V(G'_2)$ with $V(G'_2) \setminus V(G'_1)$ which contradicts our construction of $G_3$. 

Together with the above considerations, this completes the proof. 
\end{proof}

Before we can turn our attention to the final lemma, we state the following remark, which is a direct consequence of Lemma \ref{lem:construct_basis_by_identifying} and its proof.

\begin{remark} \label{rem:construct_basis_by_identifying}
    In the situation assumed by Lemma \ref{lem:construct_basis_by_identifying}, we can consider $G_1$ and $G_2$ as subgraphs of $G_3$ by identifying them with their isomorphic copies $G'_1$ and $G'_2$ which were considered in the proof of the lemma. Furthermore, the proof shows that each non-trivial block in $G_3$ is either a non-trivial block in $G_1$ or a non-trivial block in $G_2$.
\end{remark}

\begin{lemma}\label{lem:modify2}
Let $k \geq 2$ and let $H$ be a graph with $k$-basis $G$, $\chi(G) \leq k$ and $k$-representation $\mathfrak{b}: A\rightarrow \mathcal{P}(B)$. Assume that each block in $G[B]$ contains at most $k$ vertices. Let $C$ be a connected component of $G[B]$. Let $D_1, \ldots, D_l$ be the connected components of $G[V(G) \setminus V(C)]$. Let $v\in N_G(V(C))\cap V(D_1)$ and assume $\mathfrak{b}(v) \cap V(D_1) \not= \emptyset$.  Then there is a graph $H^{-}$ with $k$-basis $G^{-}$ and $k$-representation $\mathfrak{b}^{-}: A^{-} \rightarrow \mathcal{P}(B^{-})$ with the following properties:
\begin{enumerate}[label=(\roman*)]
\item \label{itm:modify2_1} $|A^{-}| < |A|$.
\item \label{itm:modify2_2} Each block in $G^-$ has at most $k$ vertices.
\item \label{itm:modify2_3} $\chi(H) \leq \max(\chi(H^{-}), k)$.
\end{enumerate}
\end{lemma}

\begin{proof} Note that $v$ as stated in Lemma \ref{lem:modify2} must be contained in $A$, and recall that by Lemma \ref{lem:neighbor_of_component}, $v$ has a unique neighbor $x_1$ in $C$, which implies the existence of edge $e=\{v,x_1\}$. Moreover, as Lemma \ref{lem:modify2} assumes $\mathfrak{b}(v)\cap V(D_1) \neq \emptyset$, we can choose some $x_2\in \mathfrak{b}(v)\cap V(D_1)$.

We will subsequently construct $(G^-, A^-, B^-, \mathfrak{b}^-)$ and then show all desired properties. At first we construct the $k$-basis $G^-$. As we conclude from Lemma \ref{lem:bridges_in_basis} that $e$ is a bridge, we observe that $G-e$ has two connected components, namely $D_1$ and $G_2:=G[V(H)\setminus V(D_1)]$. Now we apply Lemma \ref{lem:construct_basis_by_identifying}
to $G-e$ with its subgraphs $D_1$ and $G_2$ as follows: We identify the vertices $x_1$ and $x_2$ as defined above to derive a new vertex $x$ and to obtain graph $G^-$. By Lemma \ref{lem:construct_basis_by_identifying} we know that each block in $G^-$ has at most $k$ vertices. Note that this already shows Property \ref{itm:modify2_2} of Lemma \ref{lem:modify2}. We denote the vertex obtained by identifying $x_1$ with $x_2$ by $x$. Figure \ref{fig:modify2} shows an example for this construction.

Now we construct a $k$-representation $\mathfrak{b}^-: A^- \rightarrow \mathcal{P}(B^-)$. Set $A^-:= A\setminus \{v\}$ (which already shows $|A^-| < |A|$ as required by Condition  \ref{itm:modify2_1} of the lemma). Moreover, set $B^-:= \left(B\setminus \{x_1, x_2\}\right) \cup \{v, x\}$. For $a\in A^-$ set $\mathfrak{b}^{-}(a):= \left(\mathfrak{b}(a)\setminus \{x_1, x_2\}\right) \cup \{x\}$, if $\{x_1, x_2\} \cap \mathfrak{b}(a) \not= \emptyset$, and $\mathfrak{b}^-(a) := \mathfrak{b}(a)$, else. 

Now it remains to show that $(G^-, A^-, B^-, \mathfrak{b}^-)$ fulfills the conditions of Definition \ref{def:basis+representation} and Property \ref{itm:modify2_3} as stated in the lemma.

For Condition \ref{itm:condition1} of Definition \ref{def:basis+representation}  we have to check that for each non-trivial block $Z$ in $G^-$ it is true that $A^- \cap V(Z) = \emptyset$. Therefore,  recall that $G^-$ was constructed using Lemma \ref{lem:construct_basis_by_identifying} by identifying some vertex of $D_1$ and some vertex of $G_2$. By Remark \ref{rem:construct_basis_by_identifying} we know that $G^-$ contains $D_1$ and $G_2$ as subgraphs. Furthermore, again by Remark \ref{rem:construct_basis_by_identifying}, we know that $Z$ is contained either in $D_1$ or in $G_2$. Hence $Z$ is contained in $G$. Using Condition \ref{itm:condition1} for $(G, A, B, \mathfrak{b})$ we conclude that $V(Z) \cap A = \emptyset$ in $G$. With $A^-\subseteq A$ and $x\in B^-$ we derive $V(Z) \cap A^- = \emptyset$ in $G^-$ as desired.

Now let $a\in A^-$. For the remaining conditions we have to show the following two inequalities:
\begin{itemize}
    \item $\deg_{G^-}(a) \geq 2$, 
    \item $\deg_{G^-}(a) + |\mathfrak{b}^-(a)| \leq k$. 
\end{itemize}

We start with considering $\deg_{G^-}(a) \geq 2$. If $a\in A^-\subseteq A$, we have $a\notin \{v,x_1, x_2\}$ (as $a\in V(G^-)$ but $x_1, x_2\not\in V(G^-)$, and as $v\in B^-$). Deleting $e$ and identifying $x_1$ with $x_2$ leaves all degrees of vertices in $V(G) \setminus \{v, x_1, x_2\}$ unchanged (note that with the possible exception of $v$ there is no vertex which is adjacent with both $x_1$ and $x_2$ in $G$, because otherwise $e$ would not be a bridge). This implies $\deg_{G^-}(a) = \deg_G(a)$. From this and Condition \ref{itm:condition2} for $(G, A, B, \mathfrak{b})$ we conclude that $\deg_{G^-}(a) \geq 2$, which completes the proof of the first inequality.

Next we want to show $\deg_{G^-}(a) + |\mathfrak{b}^-(a)| \leq k$. We now already know  $\deg_{G^-}(a) = \deg_G(a)$. Furthermore,  our definition of $\mathfrak{b}^-$ implies $|\mathfrak{b}^-(a)| \leq |\mathfrak{b}(a)|$. Together, these observations together with Condition \ref{itm:condition3} for $(G,A,B,\mathfrak{b})$ immediately imply desired inequality, so we indeed have $\deg_{G^-}(a) + |\mathfrak{b}^-(a)| \leq k$.

So now it only remains to show $\chi(H) \leq \max(\chi(H^{-}), k)$. Let $f': V(H^-) \rightarrow \{1, \ldots, l\}$ be some $l$-coloring and let $r = \max(l,k)$. We now provide a construction of an $r$-coloring $f: V(H) \rightarrow \{1, \ldots, r\}$, which will complete our proof (by applying this construction to the special case $l=\chi(H^-)$).

We start by constructing a partial coloring $\widetilde{f}: V(H-v)\rightarrow \{1, \ldots, l\}$ as follows: 
\begin{equation*} 
        \widetilde{f}(y)=
        \begin{cases} f'(y) & \text{if $y\in V(H)\setminus \{v, x_1, x_2\}$},
        \\f'(x)  & \text{if $y\in \{x_1, x_2\}$.}
        \end{cases}
\end{equation*}
We claim that $\widetilde{f}(y) \neq \widetilde{f}(z)$ for $e=\{y,z\}\in E(H-v)$. We prove this assertion by distinguishing two cases.

\begin{enumerate}
    \item First consider the case in which we have $e\in E(G)$. We remember that $G^-$ is obtained from $G-e$ by the Zykov operation of identifying $x_1$ and $x_2$. Now, if we denote by $g$ the restriction of $f'$ to $G^--v$ and with $\widehat{g}$ the function obtained from $g$ by the construction described in Lemma  \ref{lem:zykov2}, we derive  $\widetilde{f}=\widehat{g}$. From Lemma \ref{lem:zykov2} we can thus conclude that $\widetilde{f}$ is an $l$-coloring of $G-v$, because $g$ as a restriction of $f$ is an $l$-coloring of $G^--v$. Together with $v\notin \{y,z\}$, this shows $\widetilde{f}(y) \neq \widetilde{f}(z)$. 
    \item Now consider the case in which $e\not\in E(G)$, i.e., the case in which $e$ is represented by $\mathfrak{b}$. Then we can assume without loss of generality that $y\in A$ and $z\in \mathfrak{b}$. As we have assumed $v\notin \{y,z\}$, we derive the stronger statement $y\in A^-$. Using the definition of $\mathfrak{b}^-$, we conclude the existence of some $\{y,z'\}\in E(H^-)$ such that $\{y,z'\}$ is represented by $\mathfrak{b}^-$. Note that $z' = z$ if $z\not\in \{x_1, x_2\}$, else $z'=x$. In each case we have $\widetilde{f}(z) = f'(z')$. Furthermore, $\widetilde{f}(y) = f'(y)$, because $x_1, x_2\in B$ and $y\in A$. Now we use the fact that $f'$ is an $l$-coloring of $H$ to derive the desired conclusion that $\widetilde{f}(y)\neq \widetilde{f}(z)$. 
\end{enumerate}

It remains to extend $\widetilde{f}$ to an $r$-coloring $f$ of $H$. We note that $v$ has neighbors $x_1, x_2$ with $\widetilde{f}(x_1) = \widetilde{f}(x_2)$. Moreover by Condition \ref{itm:condition3}, which holds for $(G, A, B, \mathfrak{b})$ representing $H$, we have $\deg_H(v) \leq k$. This implies that $\widetilde{f}$ uses at most $k-1$ colors in the neigborhood of $v$. Hence we need at most $k$ colors for an extension of $\widetilde{f}$ to a coloring of $H$. As we have shown that $\widetilde{f}$ is an $l$-coloring, we conclude that at most $r = \max(k,l)$ colors are needed to construct a coloring of $H$. This completes the proof.
\end{proof}

\begin{figure}
    \centering
\includegraphics{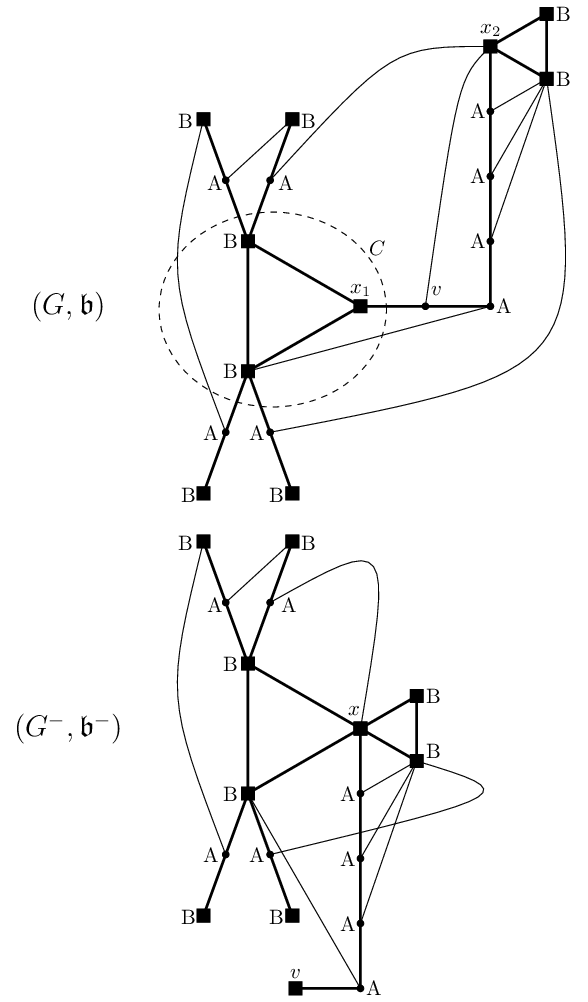}

\caption{This is an example for the construction described in the proof of Lemma \ref{lem:modify2}. We have $x_1, x_2\in B$ and $v\in A$ in $G$. $G^-$ is then obtained from $G$ by identifying $x_1$ with $x_2$ to derive the new vertex $x$. In $G^-$, vertices $x$ and $v$ are both in $B$.}
\label{fig:modify2}
\end{figure}

\subsubsection{Proof of Theorem \ref{thm:main}}
\label{sec:proof}

Before we finally prove Theorem \ref{thm:main}, note the connection of the reductions proven in the preceding subsection to said theorem: Theorem \ref{thm:main} states that for each graph $H$ which has a $k$-basis $G$ with $\chi(G)\leq k$, we also have $\chi(H)\leq k$. This implies that Theorem \ref{thm:main} can also be considered as a reduction -- in fact, it reduces the problem of coloring a graph using only $k$ colors to the problem of coloring its $k$-basis if it exists. In particular, if you want to find a coloring of a strongly tree-based-network, this is a powerful tool as trees can generally be colored using only two colors. Hence, the problem of coloring strongly tree-based networks can be reduced to the easily solvable problem of coloring trees.

Let us now turn our attention to the proof of Theorem \ref{thm:main}. We have already seen that, given a graph $H$ which is represented by a $k$-basis $G$ and $k$-representation $\mathfrak{b}: A\rightarrow \mathcal{P}(B)$, under certain conditions we can use Lemma \ref{lem:modify1} or Lemma \ref{lem:modify2} to construct a graph $H^-$ which is represented by a $k$-basis $G^-$ and $k$-representation $\mathfrak{b}^-: A^-\rightarrow \mathcal{P}(B^-)$ such that $|A^-| < |A|$. Moreover if $\chi(H^-) \leq k$, then $\chi(H) \leq k$. We will use these insights in order to derive a proof of Theorem \ref{thm:main}.

\begin{proof}[Proof of Theorem \ref{thm:main}]

 Let $H$ be a graph with $k$-basis $G$ and $k$-representation $\mathfrak{b}: A\rightarrow \mathcal{P}(B)$ such that $\chi(G) \leq k$. We now prove the desired result by induction on $|A|$. For $A = \emptyset$ the statement is clearly true, because in this case $G = H$, and by assumption we have  $\chi(G) \leq k$. 

So now we assume $A \neq \emptyset$. Moreover, without loss of generality we can assume by Proposition \ref{prop:shrinking} that each connected component of $G[B]$ is a clique of size at most $k$.

We want to find a connected component of $G[B]$ to which Lemma \ref{lem:modify1} or Lemma \ref{lem:modify2} can be applied. In the following, we will call a connected component $C$ of $G[B]$ a \emph{leaf component} if $|N_G(C)| \leq 1$. We now argue that for each connected component $C'$ of $G$ there is at least one leaf component $C$ of $G[B]$ such that $C$ is a subgraph of $C'$. To see this, contract every connected component $D$ of $G[B]$ to a vertex $v_D$. Using Lemma \ref{lem:neighbor_of_component}, we can conclude that no edge from $E_G(A, B)$ is deleted during these contractions. Furthermore, by Lemma \ref{lem:bridges_in_basis}, each edge $e\in E(G)$ which is not contained in $G[B]$ is a bridge. Hence, through the contractions we get a forest $F$ with a natural bijection between $E(F)$ and $E(G) \setminus E(G[B])$. Moreover, each connected component of $F$ resulted from some connected component $C'$ of $G$ by contracting all edges in $E(C')\cap E(G[B])$, and every vertex in $F$ is either in $A$ or was obtained by contraction from some connected component of $G[B]$. Furthermore,  $A\subseteq V(F)$ and $\deg_F(a) = \deg_G(a)\geq 2$ for every $a\in A$. Therefore, if $x$ is a leaf in $F$, then $x$ was obtained by contraction from some connected component $C$ of $G[B]$. But then $C$ is a leaf component. This shows our claim. 
 Additionally, we note that in the case that $C$ is contained in some connected component $C''$ of $G$ with $A\cap V(C'') \neq \emptyset$, we have $|N_G(C)| = 1$. So we can assume from now on that each connected component of $G$ contains at least one leaf component as a subgraph.

Now, we consider two cases. Let us first consider the case that $G$ is connected. In this case, let $C$ be a leaf component of $G[B]$, then. As we have already excluded the case $A = \emptyset$, we can conclude from the previous remarks that $|N_G(C)| = 1$. Hence, $N_G(C) = \{v\}$ for some $v\in A$. 

In this situation there is only one connected component $D_1$ of $G[V(G) \setminus V(C)]$, so we conclude that there are only two possibilities: 
\begin{enumerate}
    \item If we have $\mathfrak{b}(v) \subseteq V(C)$, then Lemma \ref{lem:modify1} can be applied as each connected component of $G[B]$ is a clique with at most $k$ elements.
    
    \item If $\mathfrak{b}(v)\cap V(D_1) \not= \emptyset$, then by Condition \ref{itm:condition1} of Definition \ref{def:basis+representation} we know that each non-trivial block in $G$ is contained in a connected component of $G[B]$. As by assumption each connected component of $G[B]$ has at most $k$ vertices, this shows that the conditions of Lemma \ref{lem:modify2} are fulfilled. 
\end{enumerate}    
In each case by application of Lemma \ref{lem:modify1} or Lemma \ref{lem:modify2}, respectively, we obtain a graph $H^-$ with $k$-basis $G^-$ and $k$-representation $\mathfrak{b}^-: A^- \rightarrow \mathcal{P}(B^-)$. Moreover, Lemma \ref{lem:modify1} and Lemma \ref{lem:modify2} tell us that in each case each block in $G^-$ has at most $k$ elements and $|A^-| < |A|$. By Corollary \ref{cor:blockcolors} we conclude that $\chi(G^-)\leq k$. As $|A^-| < |A|$ we can apply the induction hypothesis and conclude $\chi(H^-) \leq k$. Finally, from Lemma \ref{lem:modify1} and Lemma \ref{lem:modify2} we know that $\chi(H^-)\leq k$ implies $\chi(H)\leq k$, which completes the proof in this case.

It remains to consider the case in which $G$ has more than one connected component. We show how this case can be reduced to the previous case. So let us assume that $G$ has the connected components $C_1, \ldots, C_l$ with $l\geq 2$. We have already seen that each connected component $C_i$ contains at least one leaf component. In particular, this implies that $V(C_i) \cap B \neq \emptyset$, so we can choose some $b_i\in V(C_i)\cap B$ for each $i\in \{1, \ldots, l\}$. 

In order to reduce this case to the previous case we construct $(G', A', B', \mathfrak{b}')$ in the following way:
\begin{itemize}
    \item $G'$ is the graph with $V(G') = V(G)$ and $E(G') = E(G) \cup \{\{b_i, b_{i+1}\}: 1\leq i\leq k-1\}$.
    \item $A' = A$, $B' = B$ and $\mathfrak{b}'(a) = \mathfrak{b}(a)$ for all $a\in A$.
\end{itemize}
Clearly, $(G', A', B', \mathfrak{b}')$ represents some graph $H'$ with $V(H') = V(H)$ and $E(H') = E(H) \cup \{\{b_i, b_{i+1}\}: 1\leq i\leq k-1\}$. We claim that $G'$ is a $k$-basis and $\mathfrak{b}'$ is a $k$-representation of $H'$. 

All additional edges of $E(G') \setminus E(G)$ are elements of $\binom{B}{2}$. Hence, $\deg_{G}(a) = \deg_{G'}(a)$ for all $a\in A$. Moreover, $\mathfrak{b}(a) = \mathfrak{b}'(a)$ for all $a\in A$. This suffices to show Conditions \ref{itm:condition2} and \ref{itm:condition3} of Definition \ref{def:basis+representation}. In order to prove Condition \ref{itm:condition1}, we observe that each $\{b_i, b_{i+1}\}\in E(G')\setminus E(G)$ is a bridge: Indeed, $G'\setminus \{b_i, b_{i+1}\}$ is decomposed into the two connected components $\bigcup_{j=1}^i C_j$ and $\bigcup_{j=i+1}^l C_j$. Now, consider some non-trivial block $Z$ in $G'$. As all elements of $E(G')\setminus E(G)$ are bridges, we conclude that $E(Z) \cap (E(G')\setminus E(G)) = \emptyset$, and hence $Z$ is contained in $G$. Then, we conclude from Condition \ref{itm:condition1} for $(G, A, B, \mathfrak{b})$ that $A\cap V(Z) = A'\cap V(Z) =\emptyset$.

Now we have constructed a $k$-basis $G'$ of $H'$ such that $G'$ is connected. As $H$ is a subgraph of $H'$, we know  that $\chi(H) \leq \chi(H')$. As the additional edges in $E(G')\setminus E(G)$ are bridges, it is easy to conclude that $\chi(G') = \chi(G)$. 
For a successful reduction to the first case, it remains to apply Proposition \ref{prop:shrinking} to $H'$ in order to ensure that each connected component of $G'[B']$ is a clique with at most $k$ vertices. Note that Proposition \ref{prop:shrinking} guarantees that $G'$ is still connected after replacing all connected components of $G'[B']$ with cliques. This implies that we have successfully reduced the second case to the previous case, which  finally completes the proof.
\end{proof}

We already have mentioned that Theorem \ref{thm:main} together with Proposition \ref{prop:treebased_network} implies Corollary \ref{cor:treebased_network}. Thus, this solves Question \ref{qu:coloring}, which was the main motivation of our manuscript.

\section{Discussion and outlook}\label{sec:disc}

Trees have chromatic number $2$ and we have shown that strongly tree-based networks have chromatic number at most $3$. There are many other graph properties which could be investigated for the comparison of trees and strongly tree-based networks or other restricted classes of tree-based networks. The possibilities are endless here. We will give just two simple examples.

For example, if we denote by $\omega(G)$ the size of a largest clique in a graph $G$,  it is easy to see that $\omega(T) = 2$ for each tree $T$ with at least one edge. As $\omega(G) \leq \chi(G)$ is true for each graph $G$ it is implied by Theorem \ref{thm:main} that $\omega(N) \leq 3$ for each strongly tree-based network $N$. On the other hand, it has been already noted by Hendriksen that general tree-based networks can have an arbitrarily large clique number (cf. Proof of Theorem 4.1. in \cite{Hendriksen2018}). So the clique number exhibits a similar pattern as the chromatic number: Both graph properties can be arbitrarily large for general tree-based networks and have a constant as an upper bound for strongly tree-based networks.

We are curious if this pattern can be observed for more graph properties. Consider for example the \emph{list chromatic number}: Given a graph $G = (V,E)$ and a non-negative integer we associate with each vertex $v$ a list $L(v)$ of $k$ colors. Then $G$ is called \emph{$k$-list-colorable} if for each such association there is some function $f: V(G) \rightarrow L(v)$ with $f(v) \neq f(w)$ for all $\{v,w\}\in E(G)$. The list chromatic number of $G$ which is denoted by $ch(G)$ is then the smallest number $k$ for which $G$ is $k$-list-colorable. It is well-known that $ch(G) \geq \chi(G)$ \cite{Diestel2017} and that for each natural number $k$ there is a complete bipartite graph $G$ with $ch(G) > k$ \cite{Gravier1996}. Hence it is clear that the list chromatic number can be arbitrarily large for general tree-based networks. This immediately leads to the following question, which we leave open for future research.

\begin{problem} \label{pro:list_coloring}
    Is there a constant $C$ such that $ch(N) < C$ for each strongly tree-based network?
\end{problem}

Moreover, there are several coloring problems in graph theory which could be attacked by our approach. A very comprehensive overview of coloring problems can be found in \cite{Jensen1994}. Many of these problems should not be too hard to solve in the case of trees or forests. Hence, we suggest the following approach: 
\begin{enumerate}
\item Check whether the problem is easily solvable in the case of trees or forests (or another special case). 
\item Consider some generalization of the simple case which is specified in a way which is similar to Definition \ref{def:basis+representation}, i.e., consider graphs which are obtained from graphs belonging to the easiest case by adding edges in an appropriate way. 
\item Use reductions similar to the ones described in Lemma \ref{lem:shrinking}, Lemma \ref{lem:modify1} and Lemma \ref{lem:modify2} which are based on Zykov operations. This way, you can reduce the problem of coloring graphs of the general class to the problem of coloring graphs of the simple cases.
\end{enumerate}

Note that it is unlikely that the implementation of this approach will use concepts identical to the ones which were used throughout our paper. Definition \ref{def:basis+representation} has been developed especially in order to solve Question \ref{qu:coloring}. For other problems, concepts tailored to the specific question might be needed. However, we expect that our approach can be adjusted and lead to progress concerning open  coloring problems. To the best of our knowledge, there is no other example in the graph theoretical literature in which this approach has been implemented. In fact, Zykov operations have so far only been used in practical implementations for computing the chromatic number of single instances, for example in connection with brand-and-bound methods \cite{Mehrotra1996,Held2012}, but not in finding explicit values for the chromatic number of entire classes of graphs and thereby solving questions of theoretical interest.

\color{black}
Last, we are confident that the investigation of problems like the one presented in this manuscript will deepen the links between classical graph theory and tree-based networks and thus lead to a more profound understanding of the latter.

\section*{Statements and declarations}

\subsection*{Acknowledgements} The authors wish to thank three anonymous reviewers for helpful comments.

\subsection*{Competing interests} The authors herewith certify that they have no affiliations with or involvement in any organization or entity with any financial (such as honoraria; educational grants; participation in speakers’ bureaus; membership, employment, consultancies, stock ownership, or other equity interest; and expert testimony or patent-licensing arrangements) or non-financial (such as personal or professional relationships, affiliations, knowledge or beliefs) interest in the subject matter discussed in this manuscript.

\subsection*{Data availability statement} 
Data sharing is not applicable to this article as no new data were created or analyzed in this study.

\bibliography{references}

\end{document}